\documentclass[11pt,a4paper]{paper}

\usepackage{graphicx}
\usepackage{enumerate}
\usepackage{amsmath,amsthm,amsfonts}
\usepackage{microtype}
\usepackage{subcaption}
\usepackage{wrapfig}
\usepackage[a4paper]{geometry}
\usepackage[T1]{fontenc}

\usepackage{lmodern}

\graphicspath{{./pictures/}}

\usepackage[usenames,dvipsnames]{xcolor}
\usepackage[noend,linesnumbered,boxed]{algorithm2e}
\SetAlCapSkip{0.2cm}

\usepackage{cite}

\def\bd{{\partial}}

\newtheorem{theorem}{Theorem}[section]

\newtheorem{definition}[theorem]{Definition}
\newtheorem{lemma}[theorem]{Lemma}

\newcommand{\mathset}[1]{\ensuremath {\mathbb {#1}}}
\newcommand{\N}{\mathset {N}}

\newcommand{\R}{\mathset{R}}

\newcommand\C{{\mathcal C}}

\DeclareMathOperator{\diam}{diam}
\DeclareMathOperator{\dist}{d}
\DeclareMathOperator{\NIL}{NIL}
\DeclareMathOperator{\PD}{PD}

\DeclareMathOperator{\NN}{NN}

\DeclareMathOperator{\Q}{\mathcal{Q}}
\DeclareMathOperator{\wspdone}{wspd1}
\DeclareMathOperator{\wspdtwo}{wspd2}

\usepackage{ifdraft}

\title{Spanners for Directed
  Transmission Graphs\footnote{This work is supported in part by GIF
  project 1161,  DFG project MU/3501/1 and ERC StG 757609. A
  preliminary version appeared as
  Haim Kaplan, Wolfgang Mulzer, Liam Roditty, and Paul Seiferth.
\emph{Spanners and Reachability Oracles for Directed Transmission Graphs.}
Proc. 31st SoCG, pp.~156--170.}}

\author{Haim Kaplan\thanks{School of Computer Science, Tel Aviv University,
 Israel, \texttt{haimk@post.tau.ac.il}} \and
 Wolfgang Mulzer\thanks{Institut f\"ur Informatik,
Freie Universit\"at Berlin,
  Germany
  \texttt{\{mulzer,pseiferth\}@inf.fu-berlin.de}} \and
Liam Roditty\thanks{Department of Computer Science, Bar Ilan University,
  Israel
  \texttt{liamr@macs.biu.ac.il}} \and
Paul Seiferth\footnotemark[3]}

\begin{document}
\maketitle

\begin{abstract}
Let $P \subset \R^2$ be a planar $n$-point set
such that each point $p \in P$ has an
\emph{associated radius} $r_p > 0$.
The \emph{transmission graph} $G$ for $P$ is the
directed graph with vertex set
$P$ such that for any $p, q \in P$, there is
an edge from $p$ to $q$ if and only if
$d(p, q) \leq r_p$.

Let $t > 1$ be a constant. A \emph{$t$-spanner} for $G$ is
a subgraph $H \subseteq G$ with vertex set $P$ so that
for any two vertices $p,q \in P$, we have
$d_H(p, q) \leq t d_G(p, q)$, where $d_H$ and
$d_G$ denote the shortest path distance in $H$
and $G$, respectively (with Euclidean edge lengths).
We show how to compute
a $t$-spanner for $G$ with $O(n)$ edges in
$O(n (\log n + \log \Psi))$ time, where $\Psi$
is the ratio of the largest and smallest radius of a point in $P$.
Using more advanced data structures, we obtain a
construction that runs in $O(n \log^5 n)$ time, independent
of $\Psi$.

We  give two applications for our spanners.
First, we  show how to use our spanner
 to find a BFS tree in $G$ from any given start vertex
in $O(n \log n)$ time (in addition to the time it takes to build
the spanner). Second, we
show how to use our spanner to extend a reachability oracle to
answer geometric reachability queries.
In a \emph{geometric reachability query} we ask whether a
vertex $p$ in $G$ can ``reach''
a target $q$ which is an arbitrary point in the plane (rather
than restricted to be another vertex $q$ of $G$ in a
standard reachability query).
 Our spanner allows the reachability oracle  to
answer geometric reachability queries with an additive overhead
of $O(\log n\log \Psi)$ to
the query time and $O(n \log \Psi)$ to the space.
\end{abstract}

\section{Introduction}
A common model for wireless sensor networks is the \emph{unit-disk graph}:
each sensor $p$ is modeled by a unit disk centered at $p$, and there
is an edge between two sensors if and only if their disks
intersect~\cite{Clark90}.
Intersection graphs of disks with arbitrary radii have also been
used to model sensors with different
transmission strengths~\cite[Chapter~4]{Boukerche08}.
Intersection graphs of disks are undirected. However,
for some networks we may want a directed model. In such networks, a
sensor $p$ that can transmit information to a sensor $q$ may not be
able to  receive information from $q$. This motivated various
researchers to consider what  we call here
\textit{transmission graphs}\cite{RickenbachEtAl09,PelegRoditty10}.
A transmission graph $G$ is defined for
a set $P$ of points where
each point $p \in P$ has a (transmission) radius $r_p$ associated
with it.
Each vertex of $G$ corresponds to a point of $P$, and there is
a directed  edge from $p$ to $q$ if and only if $q$ lies in the
disk $D(p)$ of radius $r_p$ around $p$.
We weight each edge $pq$ of $G$ by the distance between $p$
and $q$, denoted by $|pq|$.

As many other kinds of geometric intersection graphs,
a transmission graph
may be dense and may contain $\Theta(n^2)$ edges.
Thus, if one applies a standard graph algorithm, like breadth first 
search (BFS),
to a dense transmission graph, it
runs slowly, since it requires an explicit representation of all 
the edges in the graph.
For some applications a sparse approximation of $G$ that
preserves distances suffices.
Therefore, given a transmission graph $G$, implicitly
represented by a list of points and their associated radii,
it is desirable to construct a sparse approximation of $G$ that
preserves its
connectivity and proximity properties. We want to construct
this approximation efficiently, without generating an
explicit representation of $G$.

For any $t > 1$, a  subgraph $H$ of $G$ is a \emph{$t$-spanner} for $G$ if
the distance between any pair of vertices $p$ and $q$ in $H$ is at most $t$
times the distance  between  $p$ and $q$ in $G$,
i.e., $d_H(p,q) \le t\cdot d_G(p,q)$ for any pair $p,q$
(see \cite{NarasimhanSmid07} for an overview of spanners for geometric graphs).
F\"urer and Kasivisawnathan show how
to compute a $t$-spanner
for unit- and general disk
graphs that are variations of the Yao graph \cite{FuererKasiviswanathan12, Yao82}.
Peleg and Roditty \cite{PelegRoditty10} give a construction
for $t$-spanners in transmission graphs in any metric space with bounded
doubling
dimension.
We continue these studies by giving
an almost linear time algorithm that constructs a $t$-spanner
of a transmission graph of a planar set of points
 ($P \subset \R^2$) in which the edges are weighted according 
 to the Euclidean metric
(i.e.\ $|pq|$ is the Euclidean distance between $p$ and $q$).

Our construction is also based on the \emph{Yao
graph}\cite{Yao82}.
The basic Yao graph is a $t$-spanner for the complete
graph defined by $n$ points
in the plane (with Euclidean distances as the weights of the
edges).
To determine the points adjacent to a particular
point $q$, we divide the plane by equally spaced rays
emanating from $q$ and connect $q$ to its closest point in
each wedge
(the number of wedges increases as $t$ gets smaller).
Adapting this construction to transmission graphs poses a severe
computational difficulty, as
we want to consider, in each wedge, only the points $p$
with $q \in D(p)$ and to pick the closest point to $q$ only
among those.
Since finding the exact closest point turns out to be difficult, we
need to relax this requirement in a subtle way,
without hurting the approximation too much.
This makes it possible to construct the spanner efficiently.

Even with a good $t$-spanner at hand,
we sometimes wish to obtain exact solutions for certain problems on disk graphs.
Working in this direction,
Cabello and Jej\^ci\^c gave an $O(n\log n)$
time algorithm for
computing a BFS tree in a unit-disk graph, rooted at any given
vertex~\cite{CabelloJejcic15}.
For this, they exploited the  special
structure of the Delaunay triangulation of the disk centers.
We show that our spanner admits similar properties for transmission graphs.
As a first application of our spanner,
we get an efficient algorithm to compute a BFS tree  in a 
transmission graph rooted at any given vertex.

For another application, we consider
 \emph{reachability oracles}. A reachability oracle is a data structure that can answer
\emph{reachability queries}: given two
vertices $s$ and $t$  determine if there is a
directed path from $s$ to $t$.
The quality of a reachability oracle
is measured by its query time, its space requirement, and its preprocessing time.
For transmission graphs, we can ask for a more general
\emph{geometric} reachability query: given a vertex $s$ and \emph{any} point $q \in \R^2$, determine if
 there is a vertex $t$ such that
there is a directed path from $s$ to $t$ in $G$, and $q$ lies in the
disk of $t$. We show how to extend
any given reachability oracle to answer geometric queries with a
small \emph{additive} increase in space and query time.

\paragraph*{Our Contribution and the Organization of the Paper.}
An extended abstract of this work was presented
at the 31st International Symposium on Computational 
Geometry~\cite{KaplanMuRoSe15}. 
This abstract also discusses the problem of constructing
efficient \emph{reachability oracles} for transmission graphs. 
While we were preparing the
journal version, it turned out that a full description of 
our results would yield a large and unwieldy manuscript. Therefore,
we decided to split our study on transmission graphs into two parts,
the present paper that studies fast algorithms
for spanners in transmission graphs, and a companion paper that 
deals with the construction of efficient 
reachability oracles ~\cite{KaplanEtAl15b}.

In Section~\ref{sec:spanners}, we show how to compute, for every fixed
$t > 1$, a $t$-spanner $H$ of $G$. Our construction is quite
generic and can be adapted to several situations.
In the simplest case, if the \emph{spread} $\Phi$
(i.e., the ratio between the largest and the smallest distance in $P$)
is bounded,
we can obtain a $t$-spanner in time $O(n(\log n + \log \Phi))$
(Section~\ref{sec:spanner}).
With a little more work, we can weaken the assumption to
a bounded \emph{radius ratio} $\Psi$ (the ratio between the largest and
smallest radius in $P$), giving a running time of $O(n(\log n + \log \Psi))$
(Section~\ref{sec:spannerPsi}).
Note that a bound on $\Phi$ implies a bound on $\Psi$: let $d_{\max}$ be the
largest distance and $d_{\min}$ be the smallest distance between any pair
of distinct points in $P$. We can set all radii larger than $d_{\max}$ to
be $d_{\max}$ and all radii smaller than $d_{\min}$ to $d_{\min}/2$. This does
not change the transmission graph and we have $\Psi \leq 2 \Phi$.
Using even more advanced data structures, we can compute
a $t$-spanner in time $O(n \log^5 n)$,
without any dependence on $\Phi$ or $\Psi$ (Section~\ref{sec:spannerChan}).

In Section~\ref{sec:bfstree} we show how to
 adapt a result by Cabello and
Jej\^ci\^c~\cite{CabelloJejcic15} to  compute a BFS tree in a 
transmission graph, from any given vertex $p \in P$,
in $O(n \log n)$ time, once we have the spanner ready.

In Section~\ref{sec:oracles} we show how to
 use a spanner to extend a reachability oracle to answer
geometric reachability queries. Specifically, we show that any reachability oracle
for a transmission graph with radius ratio $\Psi$, that requires $S(n)$ space, and
answers a query in
$Q(n)$ time,  can be
extended in $O(n \log n \log \Psi)$ time, to an oracle
that
can answer geometric reachability queries,  requires
$S(n) + O(n \log \Phi)$ space, and answers a query in $Q(n) + O(\log n \log \Phi)$ time.

\section{Preliminaries and Notation}
\label{sec:prelims}
\begin{wrapfigure}{R}{0.4\textwidth}
\vspace{-20pt}
\begin{center}
\includegraphics[scale=0.6]{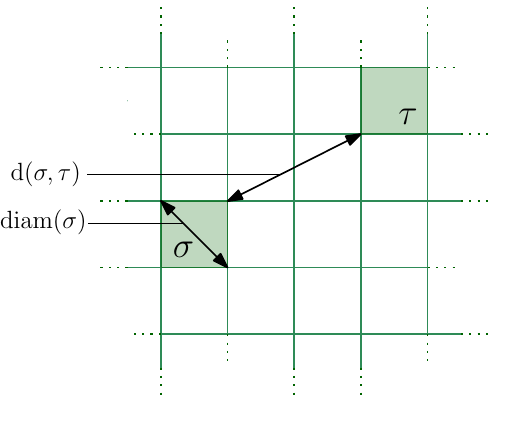}
\end{center}
\vspace{-20pt}
\caption{The grid (green) and two cells $\sigma$ and $\tau$.}
\label{fig:grid}
\end{wrapfigure}
We let $P \subset \R^2$ denote a set of $n$ points in the plane. Each
point $p\in P$ has a
\emph{radius} $r_p > 0$ associated with it.
The elements in $P$ are called \emph{sites}.
The \emph{spread} of $P$, $\Phi$, is defined as
$\Phi = \max_{p,q \in P} |pq| / \min_{p \neq q \in P} |pq|$, and
the \emph{radius ratio} $\Psi$ of $P$ is defined as
$\Psi = \max_{p, q \in P} (r_p / r_q)$. A simple
volume argument shows that $\Phi = \Omega(n^{1/2})$.
Furthermore, as stated in the introduction, we can
always assume that $\Psi \leq 2\Phi$.
Given a point $p \in \R^2$ and a radius $r$, we denote by $D(p, r)$
the closed disk with center $p$ and radius $r$. If $p \in P$, we
use $D(p)$ as a shorthand for $D(p, r_p)$. We write
$C(p, r)$ for the boundary circle of $D(p, r)$.

Our constructions make extensive use of planar grids. For $i \in \{0, 1,
\dots\}$,
we define $\Q_i$ to be the \emph{grid at level $i$}. It consists of
axis-parallel squares with diameter $2^i$ that partition the
plane in a grid-like fashion (the \emph{cells}).
We write $\diam(\sigma)$ for the diameter of a grid cell $\sigma$.
Each grid $\Q_i$ is aligned so that the origin lies at the corner of a cell.
The \emph{distance} $d(\sigma, \tau)$ between two grid cells $\sigma, \tau$ is
the smallest distance between
any pair of points in $\sigma \times \tau$, see  Figure~\ref{fig:grid}.
We assume that our model of computation allows us to find in constant time
for any given point the grid cell containing  it.

\section{Spanners for Directed Transmission Graphs}
\label{sec:spanners}
\subsection{Efficient Spanner Construction for a Set of Points with Bounded Spread}
\label{sec:spanner}

First, we give a spanner construction for the transmission
graph whose running time
depends on the spread.
Later, in Section~\ref{sec:spannerPsi}, we will tune this
construction so that the running time depends
on the radius ratio.
 The main result which we prove in this section is as follows.

\begin{theorem}
\label{thm:2dspannerSpread}
  Let $P$ be a set of $n$ points in the plane
  with spread $\Phi$.
  For any fixed $t > 1$,
  we can compute, in  $O(n\log \Phi)$ time,
  a $t$-spanner for the transmission graph $G$ of $P$.
  The construction needs  $O(n \log \Phi)$ space.
\end{theorem}
\begin{wrapfigure}{L}{0.46\textwidth}
\includegraphics[scale=1.2]{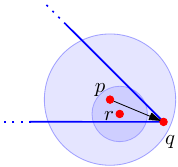}
\centering
\caption{A cone $C_q$ (blue) at a site $q$. Since $q \notin D(r)$, we pick the
edge $pq$.}
\label{fig:yao-edges}
\end{wrapfigure}
Let $\rho$ be a ray originating
from the origin and let $0 < \alpha < 2\pi$. A \emph{cone} with \emph{opening
angle} $\alpha$
and \emph{middle axis} $\rho$ is the closed region
containing $\rho$ and bounded by the two rays obtained by rotating $\rho$
clockwise and counterclockwise by $\alpha/2$.

Given a cone $C$ and a point $q \in \R^2$, we write $C_q$ for the
copy of $C$ obtained by translating the origin to $q$. We call $q$
the \emph{apex} of $C_q$.
Ideally, our spanner should look as follows. Let $\C$ be a set
of $k$ cones with opening angles $2\pi/k$
that partition the plane.
For each site $q \in P$ and each cone $C \in \C$,
we pick
the site $p \in P \cap C_q$ with $q \in D(p)$ that
is closest to $q$ (see Figure~\ref{fig:yao-edges}). We add the edge
$pq$ to $H$. The resulting graph has
$O(kn)$ edges. Using standard techniques, one can show that
$H$ is a $t$-spanner, if $k$ is large
enough as a function of $t$. This construction has been reported before and
seems to be folklore~\cite{Carmi14,PelegRoditty10}.

Unfortunately, the standard algorithms for computing the Yao graph do not
seem to adapt easily to our setting without a penalty in their running times~\cite{ChangEtAl90}.
The problem is that for each site $q$ and each cone $C_q$, we need
to search for a  nearest neighbor of $q$ only among those sites $p \in C_q$ such that
$q \in D(p)$. This seems to be hard to do with
the standard approaches.
Thus, we modify the construction to search only for an \emph{approximate} nearest neighbor of $q$ and argue
that picking an approximately shortest edge in each cone suffices to
obtain a spanner.

We partition each cone $C_q$
into ``intervals''
obtained by intersecting $C_q$ with annuli around $q$
whose inner and outer radii grow exponentially;
see Figure~\ref{fig:discretized-cone}.
There can be only $O(\log \Phi)$ non-empty intervals.
We cover each such interval
by $O(1)$ grid cells whose diameter is
``small'' compared to the width of the interval.
This gives two useful properties.
(i) We only need to consider edges from the  interval
closest to $q$ that contains sites with outgoing edges to $q$;
all other edges to $q$ will be longer.
(ii) If there are multiple edges from the same grid cell, their
endpoints are close together, and it
suffices to consider only one of them.

\begin{figure}[htbp]
 \centering
 \includegraphics[scale=0.5]{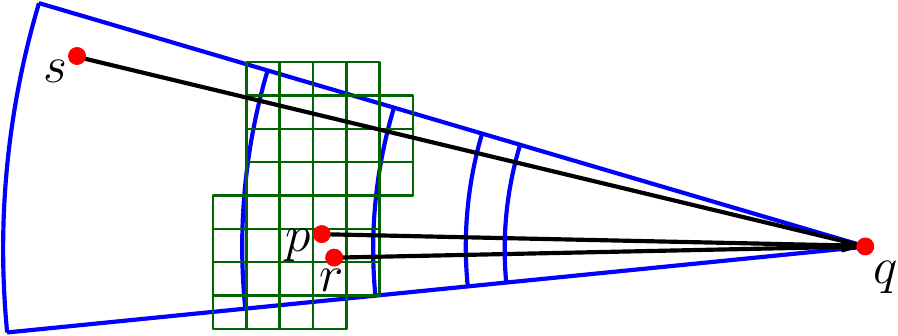}
 \caption{A cone $C_q$ covered by discretized intervals. We only need
one of the edges $pq$, $rq$ for $H$.}
 \label{fig:discretized-cone}
\end{figure}

To make this approach more concrete, we define a decomposition of $P$
into pairs of subsets of $P$ contained in certain grid cells.
These pairs represent a discretized version of the intervals
(see Definition~\ref{def:decomposition} below). This is motivated by another
spanner construction based on the \emph{well-separated pair decomposition}
(WSPD).
Let $c > 1$ be a parameter. A $c$-WSPD for $P$ is
a set of pairs $(A_i,B_i),\dots, (A_m,B_m)$ such that $A_i,B_i \subseteq
P$, and
for each pair  $a,b$ of points of $P$ there is a single index $j$ such
that $a\in A_j$ and $b\in B_j$ or vice versa. Furthermore, for any
$1 \leq
i \leq m$
 we have that $c\max\{\diam(A_i),\diam(B_i)\} \leq d(A_i,B_i)$. Here $\diam(A_i)$ is the diameter of $A_i$ and
$d(A_i,B_i)$ is the
minimum distance between any pair $a,b$ with $a \in A_i$ and $b \in B_i$.
Callahan and Kosaraju show that there always exists a WSPD with $m=O(n)$ pairs
which can be computed efficiently~\cite{CallahanKo95}.

It is well known~\cite{NarasimhanSmid07} that one can obtain a $t$-spanner
for the complete (undirected) Euclidean graph with vertex
set $P$
from a $c$-WSPD, for a large enough $c=c(t)$, by putting in the spanner an
edge $ab$ for each pair $(A_i,B_i)$ in the WSPD, where $a$ is an arbitrary point in $A_i$ and
$b$ is an arbitrary point in $B_i$.
It turns out that a similar approach works for transmission graphs. However,
since they are directed, we need to find for \emph{each} site in $B_i$
an incoming edge from a site in $A_i$, if such an edge exists, and vice versa.
This causes two difficulties: we cannot afford to check all possible edges in
$A_i \times B_i$, since this would lead to a quadratic running time, and we
cannot control the indegree of a site $p$ since it may belong to many sets $A_i$ and $B_i$.
We address the second problem by taking only $O(1)$ edges into a particular site $q$, within each of the $k$
cones of the Yao construction described above. For the first problem, we
identify in each $A_i$ a special subset  that ``covers'' all edges from a site in
 $A_i$ to a site in $B_i$, such that each
  site appears
in a constant number of such subsets.

The concrete implementation of this idea is captured by
Definition~\ref{def:decomposition}. A pair $(A_i,B_i)$ corresponds to
sets $P \cap \sigma$ and $P \cap \tau$ for two grid cell
$\sigma,\tau$ that have the same diameter and that are well separated
(Property~(i)). For a grid cell $\tau$, we denote by $m_\tau$ the site of largest radius in $P \cap
\tau$ and we define a particular subset $R_\tau\subseteq P \cap
\tau$ to be the set of sites  \emph{assigned} to $\tau$.
 Property~(ii) in Definition~\ref{def:decomposition} guarantees that each edge $pq$ of $G$ with
$q \in \sigma$ and $p \in \tau$
  is either ``represented'' in
the decomposition by an edge originating in $m_\tau$  or we have that $p\in R_\tau$.
Specifically, edges
$pq$ with $q \in P \cap \sigma$ and
$p \in P \cap \tau$  such that the disk $D(p)$ is
``large'' relative to $|pq|$ are represented by the edge $m_{\sigma}q$. This allows us
to define the sets $R_\sigma$ such that
each site appears in $O(1)$ such sets, see Figure~\ref{fig:decomposition}.
\begin{definition}
\label{def:decomposition}
Let $c > 2$ and
let $G$ be the transmission graph of a planar point set $P$.
A $c$\emph{-separated annulus decomposition} for $G$ consists
of a finite set  $\Q \subset \bigcup_{i=0}^\infty \Q_i$ of \emph{grid cells},
a symmetric \emph{neighborhood relation} $N \subseteq \Q \times \Q$ between
these
cells,
and a subset of \emph{assigned sites} $R_\sigma \subseteq P \cap \sigma$ for each
grid cell $\sigma \in \Q$.
A $c$\emph{-separated annulus decomposition} for $G$ has the following properties:
\begin{enumerate}[(i)]
\item
For every $(\sigma,\tau)\in N$, $\diam(\sigma) = \diam(\tau)$,
and
$d(\sigma,\tau) = \gamma \diam(\sigma)$, for some $\gamma \in [c-2, 2c)$.
\item for every edge $pq$ of $G$, there is a
pair $(\sigma,\tau) \in N$ with $q \in \sigma$, $p \in \tau$, and
either $p \in R_{\tau}$ or $q \in D(m_{\tau})$.
\end{enumerate}
\end{definition}
The following fact is a direct consequence of Definition~\ref{def:decomposition}.
For each cell $\sigma \in \Q$, we define its \emph{neighborhood} as
$N(\sigma) = \{\tau \mid (\sigma, \tau) \in N\}$.
\begin{lemma}\label{lem:volume}
For each cell $\sigma \in \Q$, we have
$|N(\sigma)| = O(c^2)$, and for each cell $\tau \in \Q$ the number of cells
$\sigma \in \Q$ such that $\tau \in N(\sigma)$ is $O(c^2)$.
\end{lemma}

\begin{proof}
This follows from Definition~\ref{def:decomposition}(i) via a
standard volume argument.
\end{proof}

\begin{figure}[htb]
\centering
\begin{subfigure}[b]{0.4\textwidth}
\centering
\includegraphics[width=0.7\textwidth]{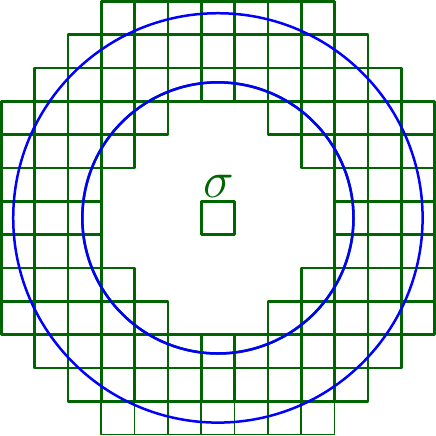}
  \caption{By Property~(i) in Definition~\ref{def:decomposition}
 $N(\sigma)$ covers an annulus.}
\end{subfigure}
\begin{subfigure}[b]{0.53\textwidth}
\centering
\includegraphics[width=0.7\textwidth]{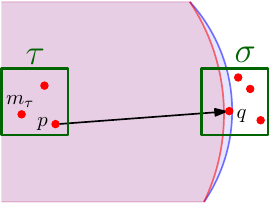}
  \caption{Since $D(m_\tau)$ (red) does not contain $q$, we need  to put $p$ in
$R_\tau$
to cover the edge $pq$ (Property~(ii)).}
\end{subfigure}
\caption{Illustration of Definition~\ref{def:decomposition}}
\label{fig:decomposition}
\end{figure}
Given this decomposition, we first present a simple (and rather inefficient)
rule for picking incoming edges such that the resulting graph is a $t$-spanner.
Then we explain how to compute the decomposition using
a \emph{quadtree}. Finally, we exploit the quadtree
to make the spanner construction efficient.

\paragraph*{Obtaining a Spanner.}
Let $t > 1$ be the desired stretch.
We  pick a suitable separation
parameter $c$ and a number of cones $k$ that depend on $t$, as specified later.
Let $(\Q, N, R_\sigma)$ be a $c$-separated annulus
decomposition for $G$.
For a cone $C \in \C$ and an integer $\ell \in \N$,
we define $C^\ell$
as the cone with the same middle axis as $C$
but with an opening angle $\ell$ times larger than the opening angle of
$C$.
For $\sigma \in \Q$,
let $C_\sigma$ be the copy of $C$ with
the center of $\sigma$ as the apex.

To obtain a $t$-spanner $H \subseteq G$,
we pick the incoming edges
for each site $q \in P$ and each cone $C \in \C$
as follows (see Algorithm~\ref{alg:edgeselection}).
We consider the cells of $\Q$ containing $q$ in increasing
order of diameter. Let $\sigma$ be one such cell containing $q$ that
we process. We traverse all neighboring cells $\tau$ of $\sigma$,
 that are contained in $C_\sigma^2$.
For each such neighboring cell $\tau$, we check if there exists a
site $r\in R_\tau \cup \{m_\tau\}$ that has an outgoing edge to $q$.
If such a site exists, we
add to $H$ an edge to $q$ from a single, arbitrary, such site $r$.
After considering \emph{all} neighbors $\tau$ of $\sigma$
we terminate the  processing
of  $q$  and $C$ if we added at least one edge incoming to  $q$.
If we have not added any edge into $q$ while processing all neighbors $\tau$
of $\sigma$ we continue
  to the next largest
cell containing $q$.  We use here
the extended cones $C_\sigma^2$ (instead of the
cone $C_q$) to gain certain
flexibility that will be useful
for later extensions of Algorithm~\ref{alg:edgeselection}.

\begin{algorithm}[htb]
$\Q_q \gets $ cells of $\Q$ that contain $q$\\
Sort the cells in $\Q_q$ in increasing order by diameter\\
Make $q$ \emph{active}\\
\While{$q$ \textnormal{is active}}{
  $\sigma \gets$ next largest cell in $\Q_q$\\
  \ForEach{\textnormal{cell }$\tau \in N(\sigma)$ \textnormal{that is
contained in} $C_\sigma^2$}{
  \label{line:conetest}
    \If{\textnormal{there is a} $r \in R_{\tau} \cup \{m_{\tau} \}$
     \textnormal{with} $q \in D(r)$} {
      \label{line:edgeselection}
      Take an arbitrary such $r$, add the edge $rq$ to $H$,
      and set $q$ to \emph{inactive}
    }
  }
}
\caption{Selecting the incoming edges for $q$ and the cone $C$.}
\label{alg:edgeselection}
\end{algorithm}

For each cone $C \in \C$ and each site $q \in P$ there is only one cell
$\sigma \in \Q_q$ that produces incoming edges for $q$.
We have $k$ cones and $|N(\sigma)| = O(c^2)$ by Lemma~\ref{lem:volume}, so
$q$ has $O(c^2k)$ incoming edges.
It follows that the size of $H$ is $O(n)$ since
 $c$ and $k$ are constants.

Next we show that $H$ is a $t$-spanner. For this,
we show that every edge $pq$ of $G$ is represented in $H$ by
an approximate path.
We prove this by induction on the ranks of the
edge lengths. This is done in a similar manner as for the standard Yao
graphs, but with a few twists that require three additional technical
lemmas.
Lemma~\ref{lem:centercone} deals with
the imprecision introduced by taking
the cone $C_\sigma^2$ instead of $C_q$. It follows from this lemma that
if $pq$ is contained in the cone $C_q$ then
Algorithm~\ref{alg:edgeselection} picks at least one edge $rq$ with
$r \in C_q^4$.
Lemma~\ref{lem:annulusDiameter} and Lemma~\ref{lem:lengthofchosenedge}
encapsulate geometric facts
that are used to bound the distance between
the endpoints $r$  and $p$ depending on whether $|rq|$ is
larger or smaller than $|pq|$.
Lemma~\ref{lem:lengthofchosenedge} is due to Bose et
al.~\cite{BoseEtAl12} and for completeness we include their proof.

\begin{lemma}
\label{lem:centercone}
Let $c > 3 + \frac{2}{\sin (\pi /k)}$ and
let $\ell \in \{1, \dots, \lfloor k/2 \rfloor\}$.
Consider a cell $\sigma \in \Q_i$ and a
cone $C \in \C$.
Fix two points $q,s \in \sigma$.
Every cell $\tau \in \Q_i$ with
$d(\sigma,\tau) \geq (c-2)2^i$ that intersects the cone $C^\ell_q$ is contained
in the cone $C_s^{2\ell}$.
In particular, any point $p \in C_q^\ell$
with $|pq| \geq (c-2)2^i$ lies in a cell that is fully contained in $C_s^{2\ell}$.
\end{lemma}
\begin{proof}
Let $x$ be a point in $\tau \cap C^{\ell}_q$.
By assumption, $|xq| \geq (c-2)2^i$.
Let $D = D(x, 2^i)$ be the disk with center $x$ and radius $2^i$.
Then, $\tau \subseteq D$. We  show that $C_s^{2\ell}$ contains $D$
and thus $\tau$.
Since $\sigma$ has diameter $2^i$, and  $C_q^\ell$ contains $x$,
the translated copy $C_s^\ell$ must
 intersect $D$. If $D \subset C_s^\ell$, we are done.
Otherwise, there is a boundary ray $\rho$ of $C_s^\ell$ that  intersects the
boundary of $D$. Let $y$ be the first intersection of $\rho$ with the boundary of $D$.
See Figure~\ref{fig:cone-imprecision}.

Since $s \in \sigma$ and $x \in \tau$,
the triangle inequality gives that $|ys| \geq |xs| - |xy| \geq  (c-3)2^i$.
Let $\rho'$ be the boundary ray of $C_s^{2\ell}$ corresponding to $\rho$ and let
$y'$
be the orthogonal projection of $y$ onto $\rho'$. Since
$|ys| \geq (c-3)2^i$ and since the angle between $\rho$ and $\rho'$ is
$\pi \ell/k$, we get that $|yy'| \geq (c-3)2^i\sin( \pi \ell/k)$.
It follows that $|yy'| \geq 2\cdot
2^i$
for $c > 3 + \frac{2}{\sin( \pi \ell/k)}$.
This holds for any  $\ell \in \{1, \dots, \lfloor k/2 \rfloor\}$ if
$c \geq 3 + \frac{2}{\sin (\pi/k)}$.
Thus, $\tau \subset D \subset C_s^{2\ell}$.
\end{proof}
\begin{figure}[htbp]
 \centering
 \includegraphics[scale=0.6]{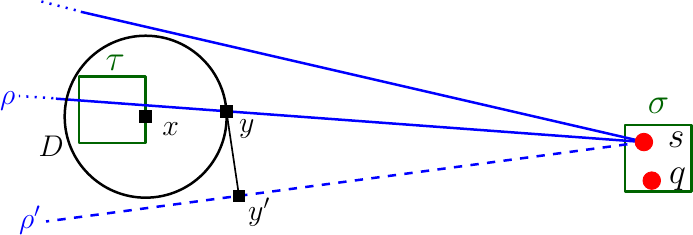}
 \caption{The boundary ray $\rho$ of $C_s^\ell$ intersects the boundary of $D$ in
$y$.}
 \label{fig:cone-imprecision}
\end{figure}

Let $p$ be a site in $C_q$ such that $pq$ is an edge of $G$, and
$p \in \tau \in N(\sigma)$ where $\sigma$ is a cell containing $q$.
Then by Lemma~\ref{lem:centercone}, $\tau$ is contained in $C_\sigma^2$.
It follows that Algorithm~\ref{alg:edgeselection} either finds an edge $rq$ before processing
$\sigma$, or finds an edge $rq$ with $r\in \tau$ while processing $\sigma$.
By applying Lemma~\ref{lem:centercone} again we get that $r \in C_q^4$.
This fact is described in greater detail and is being used in the proof of Lemma \ref{lem:shorteredge} below.

\begin{lemma}\label{lem:annulusDiameter}
Let $C \in \C$, and let $q \in \R^2$.
Suppose there are two points $p,r \in C_q^4$
with $(c-2)2^i \le |pq| \le |rq| \le (c+1)2^i$.
Then $|pr| \leq ((8\pi/k)(c+1)+3)2^i$.
\end{lemma}
\begin{proof}
The points $p$ and $r$ lie in an annulus around $q$ with inner
radius $(c-2)2^i$ and outer radius $(c + 1)2^i$.
Since $p,r \in C_q^4$, when going from $p$ to $r$, we must travel at most
$(8\pi/k) (c+1)2^i$ units along the circle around $q$ with
$p$ on the boundary, then at most $3 \cdot 2^i$ units
radially towards $r$. Thus, $|pr| \leq (8\pi/k) (c+1)2^i + 3\cdot 2^i$.
\end{proof}

\begin{figure}[htb]
\centering
\begin{subfigure}[t]{0.45\textwidth}
\centering
\includegraphics[scale=0.7]{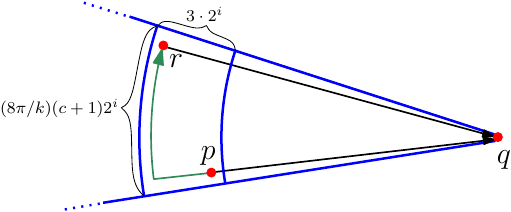}
  \caption{Lemma~\ref{lem:annulusDiameter}. Two sites in an annulus are close to each other.}
\end{subfigure}
\begin{subfigure}[t]{0.45\textwidth}
\centering
\includegraphics[scale=0.7]{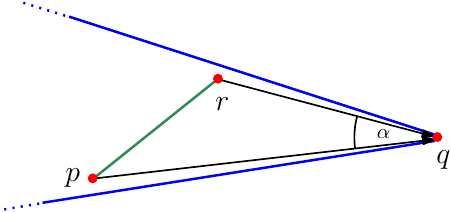}
  \caption{Lemma~\ref{lem:lengthofchosenedge}. If $\alpha$ is small
and $|rq| \leq |pq|$, then
$|pr| < |pq|$.}
\end{subfigure}
\end{figure}

\begin{lemma}[Lemma~10 in \cite{BoseEtAl12}]
\label{lem:lengthofchosenedge}
Let $k \geq 25$ be large enough such that
\[
 \frac{1 + \sqrt{2 - 2\cos (8\pi / k)}}{2\cos (8\pi/k) -1}
  = 1 + \Theta(1/k) \leq t
\]
for our desired stretch factor $t$.
For any three
distinct points $p,q,r \in \R^2$ such that
$|rq| \leq |pq|$ and $\alpha = \angle pqr$
is between $0$ and $8\pi /k$, we have
 $|pr| \leq |pq| - |rq|/t$.
\end{lemma}
\begin{proof}
By the law of cosines and since $0\le \alpha \le 8k/\pi$ we have that
\begin{align*}
|pr|^2 &= |pq|^2 + |rq|^2 -2|pq|\cdot|rq|\cos \alpha
\leq
|pq|^2 + |rq|^2 -2|pq|\cdot|rq|\cos (8\pi/k)
\intertext{Introducing $t$ by adding and subtracting equal terms,  this is}
&=
|pq|^2 - \frac{2}{t}\,|pq|\cdot|rq| + \frac{1}{t^2}|rq|^2  +
\frac{t^2 - 1}{t^2}|rq|^2 - \frac{2(t\cos (8\pi/k)- 1)}{t}|pq|\cdot|rq| \\
&=
\Bigl(|pq| - \frac{|rq|}{t}\Bigr)^2  +
\frac{t^2 - 1}{t^2}|rq|^2 - \frac{2(t\cos (8\pi/k)- 1)}{t}|pq|\cdot|rq|.
\end{align*}
We complete the proof by showing that under the assumptions of the lemma
$\frac{t^2 - 1}{t^2}|rq|^2 - \frac{2(t\cos (8\pi/k)- 1)}{t}|pq|\cdot|rq| \le 0$.
We have that
\begin{align*}
\frac{t^2 - 1}{t^2}|rq|^2 - \frac{2(t\cos (8\pi/k)- 1)}{t}|pq|\cdot|rq|
&=\frac{|rq|^2}{t^2}\Bigl(t^2 - 1 - 2(t^2\cos (8\pi/k) - t)\frac{|pq|}{|rq|}\Bigr)\\
&\leq
\frac{|rq|^2}{t^2}\Bigl(t^2 - 1 - 2(t^2\cos (8\pi/k) - t)\Bigr),
\end{align*}
where the last inequality follows
since $|pq| \geq |rq|$ and
\[
t\ge  \frac{1 + \sqrt{2 - 2\cos (8\pi / k)}}{2\cos (8\pi/k) -1} \ge \frac{1}{2\cos (8\pi/k) -1} \ge \frac{1}{\cos (8\pi/k)} \ ,
\]
so
 $t\cos(8\pi/k) \geq 1$.
Now we have that
$$t^2 - 1 - 2(t^2\cos (8\pi/k) - t)
= (1-2\cos (8\pi/k))t^2 + 2t - 1 \le 0 $$
if
 $\cos(8\pi/k) > 1/2$ and
\[
 \frac{1 + \sqrt{2 - 2\cos (8\pi / k)}}{2\cos (8\pi/k) -1} \leq t \ .
\]
The latter inequality holds by assumption and $\cos(8 \pi/k) > 1/2$ for $k \geq 25$.
\end{proof}

We are now ready to bound the stretch of the spanner $H$. This is done in two
steps.
In the first step (Lemma \ref{lem:shorteredge}) we prove that for any edge $pq$ of $G$
which is not
 in
$H$, there exists a shorter edge $rq$ in $H$,
such that $r$ is ``close'' to $p$.
 This fact allows us to prove, via a fairly standard
inductive argument, that $H$ is indeed a spanner of $G$.

\begin{lemma}
\label{lem:shorteredge}
Let $c$ and $k$ be such that $c > 3 + \frac{2}{\sin (\pi /k)}$ as
required by
 Lemma \ref{lem:centercone}, $k$ satisfies the
conditions of Lemma \ref{lem:lengthofchosenedge} and, in addition, $c \geq 2 + \frac{2t}{t-1}$ and $k \ge \frac{16\pi t}{t-1}$.
Let $pq$ be an edge of $G$. Then either $pq$ is in $H$ or there is an edge $rq$ in $H$ such that
$|pr| \leq |pq| - |rq|/t$.
\end{lemma}
\begin{proof}
Let $N$ be the neighborhood relation of the
$c$\emph{-separated annulus decomposition} used by Algorithm~\ref{alg:edgeselection}.
Let $(\sigma, \tau) \in N$ be a pair of neighboring
cells satisfying  requirement (ii) of Definition~\ref{def:decomposition} with respect to
$pq$.
In particular we have that $q \in \sigma$ and $p \in \tau$.
If there is more than one such pair $(\sigma, \tau) \in N$, we consider the pair with minimum diameter.
Let $\diam(\sigma) = 2^i$, that is
 $\sigma,\tau \in \Q_i$.

Let $C \in \C$ be the cone such that
$p \in C_q$.
Since $p \in C_q \cap \tau$ and since $d(\sigma,\tau) \geq (c-2)2^i$,
Lemma~\ref{lem:centercone} implies that $\tau \subset C_\sigma^2$. Hence,
$\tau$ is considered for incoming edges for $q$
(line~\ref{line:conetest} in Algorithm~\ref{alg:edgeselection}).
We split the rest of the proof into two cases.

\noindent
\textbf{Case 1:} $q$ remains active until $(\sigma, \tau)$
is considered.
Requirement (ii) of
Definition~\ref{def:decomposition} guarantees that
Algorithm~\ref{alg:edgeselection} finds
an incoming edge $rq$ for $q$
with $r \in \tau$.
If $r = p$, we are done,  so
suppose that $r \neq p$.
Since $\diam(\sigma) = 2^i$ and $|rq| \geq d(\sigma,\tau) \geq (c-2)2^i$
we have
\begin{align*}
  |pr| & \leq 2^i  = |pq| - (|pq| - 2^i) \leq |pq| - (|rq| - 2\cdot 2^i) \\
  & \leq |pq| - (|rq| - 2|rq|/(c-2)) \leq |pq| - |rq|(1-2/(c-2)) \leq |pq|
-|rq|/t,
\end{align*}
for $c \geq 2 + \frac{2t}{t-1}$.

\begin{figure}[htb]
\centering
\begin{subfigure}[t]{0.45\textwidth}
\centering
\includegraphics[scale=0.8]{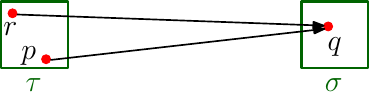}
  \caption{Case 1: $p$ and $r$ are in the
same cell $\sigma$.}
\end{subfigure}
\begin{subfigure}[t]{0.45\textwidth}
\centering
\includegraphics[scale=0.8]{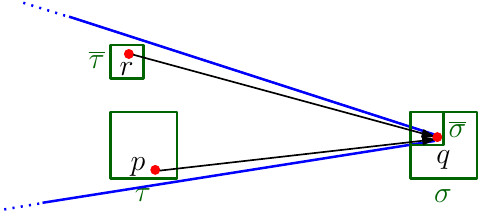}
  \caption{Case 2: $p$ and $r$ are in different cells with different levels
but in the same cone $C_q^4$.}
\end{subfigure}
\end{figure}

\noindent
\textbf{Case 2:} $q$ becomes inactive before $(\sigma, \tau)$
is considered.
Then Algorithm~\ref{alg:edgeselection} has selected an edge $rq$
while considering
a pair $(\bar{\sigma},\bar{\tau}) \in N$ with $q \in \bar{\sigma}$,
$r \in \bar{\tau}$ and $\diam(\bar{\sigma}) \leq 2^{i-1}$.
We now distinguish two subcases.

\textbf{Subcase 2a}  $|rq| \geq |pq|$.
From Property (i) of Definition~\ref{def:decomposition}, it follows that
$d(\sigma,\tau) \ge (c-2)2^i$ and therefore $|pq| \geq (c-2)2^i$.
It also follows from the same property that
$d(\bar{\sigma},\bar{\tau}) \le 2c2^{i-1}$, so
$|rq| \le 2c2^{i-1} + 2\cdot 2^{i-1} = (c+1)2^i$.
Combining these inequalities we obtain that
$(c-2)2^i \le |pq| \le |rq| \le   (c+1)2^i $
and therefore
$|pq| \geq |rq| - 3\cdot 2^i$.
Lemma~\ref{lem:annulusDiameter}
implies that
$|pr| \leq ((8\pi/k)(c+1)+3) 2^i$, and thus we have
\begin{align*}
 |pr| &\leq ((8\pi/k)(c+1)+3) 2^i \\
  &\leq  |pq| - |pq| + ((8\pi/k)(c+1)+3) 2^i \\
  &\leq |pq| - \bigl(|rq| - 3\cdot2^i - ((8\pi/k)(c+1)+3) 2^i\bigr) \\
  &\leq |pq| - \Bigl(|rq| - \frac{(8\pi(c+1)-6)2^i}{k}\Bigr)\\
  &\leq |pq| - |rq|\Bigl(1 - \frac{(8\pi(c+1)-6)}{k(c-2)}\Bigr) \\
 &\leq |pq| - |rq|\Bigl(1 - \frac{16\pi}{k}\Bigr) \ .
\end{align*}
The  third inequality follows since
$|pq| \geq |rq| - 3\cdot 2^i$ as we argued above, and  the
fifth inequality follows since
$2^i \leq |rq|/(c-2)$. The last inequality holds for $c \ge 5$ (which follows from our assumptions).
Now we clearly have that
\[
|pq| - |rq|\Bigl(1 - \frac{16\pi}{k}\Bigr)
\leq |pq|-
|rq|/t,
\]
for $k \ge \frac{16\pi t}{t-1}$.

\textbf{Subcase 2b}  $|rq| <|pq|$.
By assumption, we have $p \in C_{q} \subset C_q^4$.
Furthermore, by applying Lemma~\ref{lem:centercone} with
 the midpoint of $\bar{\sigma}$ as $q$, $r$ as $p$, and $q$ as $s$,
in the statement of the lemma, we get that
 $r \in C_q^4$.
Since $p,r \in C_q^4$ and since the opening angle of  $C_q^4$ is $8\pi/k$, it follows from
Lemma~\ref{lem:lengthofchosenedge} that
$|pr| \leq |pq| - |rq|/t$.
\end{proof}

\begin{lemma}
 \label{lem:edgeapproximation}
For any $t > 1$, there are constants
$c$ and $k$ such that $H$ is a $t$-spanner for the transmission graph
$G$.
\end{lemma}
\begin{proof}
We pick the constants $c$ and $k$ so that Lemma \ref{lem:shorteredge} holds.
We prove
 by induction on the indices of edges when ordered by their lengths,
 that for each  edge $pq$ of $G$,
there is a path from  $p$ to $q$ in $H$
of length at most $t|pq|$.
For the base case, consider the shortest edge
$pq$ in $G$.
By Lemma~\ref{lem:shorteredge},
if $pq$ is not in $H$ then  there is an edge $rq$ in $H$  such that
 $|pr| \leq |pq| - |rq|/t$. Since $pq$ is an edge of $G$, it follows that $r_p \geq |pq|$
and therefore $pr$ must  also be an edge of $G$, and it is shorter than $pq$. This gives   a contradiction and therefore
 $pq$ must be  in $H$.

For the induction step, consider an edge
$pq$ of $G$. If $pq$ is in $H$ we are done. Otherwise
by Lemma~\ref{lem:shorteredge} there is an edge $rq$ in $H$ such that
 $|pr| \leq |pq| - |rq|/t$. As argued above, $pr$ is an edge of $G$ shorter than $|pq|$ so
by the induction hypothesis, there is a path from $p$ to $r$ in $H$ of length
no larger than $t|pr|$. It follows that
\begin{equation*}
 d_H(p,q) \leq d_H(p,r) + |rq| \leq t|pr| + |rq| \leq  t(|pq| - |rq|/t) +
|rq| \leq t|pq|,
\end{equation*}
as required.
\end{proof}

\paragraph*{Finding the Decomposition.}
We use a quadtree to define the cells of the decomposition. We recall that
a \emph{quadtree}  is a rooted tree $T$ in which each
internal node has degree four.
Each node $v$ of $T$ is associated with a cell $\sigma_v$ of some grid
$\Q_i$, $i\geq 0$, and if $v$ is an internal node,
the cells associated with its children partition $\sigma_v$ into four
congruent squares, each with diameter $\diam(\sigma_v)/2$.
If $\sigma_v$ is from $\Q_i$ then we say that $v$ is of \emph{level} $i$.
Note that all nodes of $T$ at the same distance from the root are of the same
level.

Let $c$ be the required  parameter for the annulus decomposition.
We scale $P$ such that the
closest pair in $P$ has distance $c$. (We use $P$  to denote also the scaled point set).
Let $L$ be the smallest integer such that we can translate $P$ so that it fits in a
 single cell $\sigma$ of  $\Q_L$.
Since $c$ is constant and $P$ has spread $\Phi$,
the diameter of $P$ (after scaling) is $c\Phi$ and therefore $L = O(\log \Phi)$.
We translate $P$ so that it fits in $\sigma$ and we
 associate the root $r$ of our quadtree $T$ with this cell $\sigma$, i.e.\  $\sigma_r = \sigma$.
By the definition of a level, $r$ is of level $L$.

We continue constructing $T$ top down as follows.
We construct level $i-1$ of $T$, given level $i$, by splitting
the cell $\sigma_v$ of each node $v$, whose cell $\sigma_v$ is not empty,
into four
congruent squares, and associate each of these squares with a child of $v$.
We stop the construction of $T$ after generating the cells of level $0$. The scaling which we did to
$P$
 ensures that each cell of a  leaf node at level $0$
contains at most one site.

We now set $\Q = \{ \sigma_v \mid v \in T \}$.
We define $N$ as the set of all pairs
$(\sigma_v,\sigma_w) \in \Q \times \Q$ such that $v$ and $w$
are at the same level in $T$ and
$ d(\sigma_v,\sigma_w) \in [c-2,2c) \diam(\sigma_v)$.\footnote{We
denote the interval $[a\diam(\sigma_v),b\diam(\sigma_v))$ by
$[a,b) \diam(\sigma_v)$.}
For $\sigma \in \Q$, we define $R_{\sigma}$ to be the set of all sites  $p \in \sigma \cap P$
with  $ r_p \in [c,2(c+1))\diam(\sigma_v)$.

\begin{lemma} \label{lem:3.9}
$(\Q, N, R_\sigma)$
is a $c$-separated annulus decomposition for $G$.
\end{lemma}
\begin{proof}
Property (i) of
Definition~\ref{def:decomposition} follows by construction.
To prove that Property (ii) holds consider an edge
$pq$ of $G$.
Let $i$ be the integer such that  $|pq| \in [c, 2c) 2^i$.
Let $\sigma, \tau$ be the cells of $\Q_i$
with $p \in \sigma$ and $q \in \tau$.
By construction,  $\sigma$ and $\tau$ are assigned to nodes
of the quadtree and thus contained in $\Q$.
Since $\diam(\sigma) = \diam(\tau) = 2^i$, we have
\[
  (c-2)2^i \leq |pq| - 2 \diam(\sigma) \leq d(\sigma, \tau) \leq |pq| <
c2^{i+1},
\]
and therefore $(\sigma,\tau) \in N$
by our definition of $N$.
Since $pq$
is an edge of $G$, it follows that $r_p \geq |pq| \geq c2^i$.
If $r_p < (c+1)2^{i+1}$, then $p \in R_{\sigma}$.
Otherwise, $r_{m_{\sigma}} \geq r_p \ge (c+1)2^{i+1}$,
and $q \in \tau \subset D(m_{\sigma})$.
\end{proof}

\paragraph*{Computing the Edges of $H$.}
We find edges for each cone  $C \in \C$ separately as follows.
For each pair of neighboring cells $\sigma$ and $\tau\in N(\sigma)$ such that $\tau$ is contained in
$C_{\sigma}^2$ we find
all incoming edges to sites in $\sigma$ from sites in $\tau$ simultaneously.
To do this efficiently, we need to sort the sites in $\sigma$ along the
$x$ and $y$ directions. Therefore, we process the  cells
bottom-up along $T$ in order of increasing levels. This way we can obtain a sorted list of the sites in each
cell $\sigma$ by merging the sorted lists of its children.
See Algorithm~\ref{alg:efficientedgeselection}.

\begin{algorithm}[htb]
\For{$i=0,\dots, L$} {
  \ForEach{$v \in T$ \emph{of level} $i$}{
    $Q \gets $ active sites in $\sigma_v \cap P$\\ \tcp{preproccesing}
    Sort $Q$ in $x$ and $y$-direction by merging the sorted lists of the
     children of $v$
    \label{line:preproccesing}
    \ForEach{$\tau \in N(\sigma_v)$ \textnormal{contained in }
$C_{\sigma_v}^2$}{
      $R \gets R_{\tau} \cup \{m_\tau\}$\\
      \tcp{edge selection}
      For each site $q \in Q$, find
      a $r \in R$ with $q \in D(r)$,
      if it exists; add the edge $rq$ to $H$

      \label{line:edgeselection2}
    }
    Set all $q \in Q$ for which at least one
    incoming edge was found to \emph{inactive}
  }
}
\caption{Selecting the edges for $H$ for a fixed cone $C$.}
\label{alg:efficientedgeselection}
\end{algorithm}

Note that the edges selected by
Algorithm~\ref{alg:efficientedgeselection} have the same
properties as the edges selected by
Algorithm~\ref{alg:edgeselection}.
Thus, by Lemma~\ref{lem:edgeapproximation}, the
resulting graph is a $t$-spanner.
Let $Q$ be the set of active sites in $\sigma_v$ when processing $v$.
Let $\tau \in N(\sigma_v)$ such that $\tau$ is contained in $C_{\sigma_v}^2$
and let $R = R_{\tau} \cup \{m_{\tau}\}$. Assume  $|Q|=n$ and $|R|=m$. To find the edges from sites in $R$ to
sites in $Q$ efficiently, we use the fact that these sets of sites are separated by a line
parallel to
either the $x$- or the $y$-axis.

Assume without loss of generality that $\ell$ is the $x$-axis, the sites of $R$ are above
$\ell$ and the sites of $Q$ are below $\ell$, and assume that  $Q$ is sorted along $\ell$.
For each site $p\in R$ we take the part of $D(p)$ which lies below $\ell$ and compute the union of these
``caps''.
 This  union  is bounded from
above by $\ell$ and from below by the lower envelope of the arcs of the boundaries of the caps.
The complexity of the boundary of this union is $O(m)$ and it can be computed in $O(m\log m)$ time~\cite{AgarwalSharir96}. See Figure~\ref{fig:lowerenvelope}.

Once we have computed this union we check for each $q\in Q$
whether $q$ lies inside it.
This can be done by checking whether the intersection, $z$, of a vertical line through $q$ with the union is above or below $q$.
If $q$ is above $z$ then we add the edge $rq$ to $H$ where $r$ is the site such that $z\in \bd D(r)$.
We perform this computation for all sites in $Q$ together by a
 simple sweep in the $x$-direction
while traversing in parallel the lower envelope of the caps  and the sites of $Q$.
This clearly takes $O(m + n)$ time.

\begin{figure}[htb]
 \centering
\includegraphics[scale=0.6]{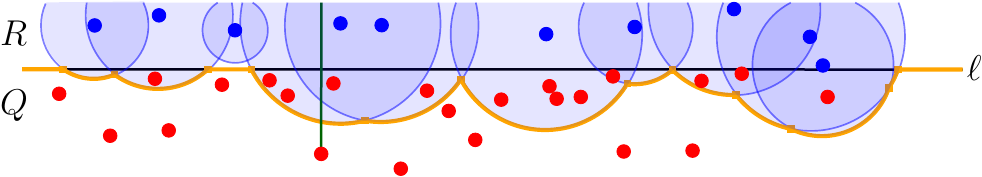}
\caption{The lower envelope (orange), the sites $Q$ (red) and $R$
(blue), and the sweepline (green).}
\label{fig:lowerenvelope}
\end{figure}

We thus proved the following  lemma.
\begin{lemma}
\label{lem:lowerenvelope}
Let $Q$, $R$, and $\ell$ be as
above with $|Q|=n$ and $|R|=m$.
Suppose that $Q$ is sorted along $\ell$ and that
$\ell$ separates $Q$ and $R$.
We can compute in $O(m \log m + n)$ time for each $q \in Q$ one
disk from $R$ that contains it, provided that such a disk exists.
\end{lemma}

\paragraph*{Analysis.}
We prove that
Algorithm~\ref{alg:efficientedgeselection}
runs in $O(n \log \Phi)$ time and uses
$O(n \log \Phi)$ space. The
running time is dominated by the edge selection step  described in
 Lemma~\ref{lem:lowerenvelope}.
We  argue that each site participates in $O(1)$ edge selection steps
as a disk center (in $R$) and in $O(\log \Phi)$
edge selection steps as a vertex looking for incoming
edges. From these observations (and the fact that
$\Phi = \Omega(n^{1/2})$) the stated time bound essentially follows.
\begin{lemma}
 \label{lem:runningtime}
We construct the spanner $H$ of the transmission graph $G$ in
$O(n\log \Phi )$ time and space.
\end{lemma}
\begin{proof}
 The quadtree $T$ can be computed in
 $O(n \log \Phi)$ time and space \cite{4M}, and within this time  bound we can also compute
$N(\sigma_v)$, $R_{\sigma_v}$, and $m_{\sigma_v}$ for
each node $v \in T$.

Merging the sorted lists of the sites in $\sigma_w$ for each child $w$ of $v$ to
obtain the sorted list of the sites in $\sigma_v$ (line~\ref{line:preproccesing} in
Algorithm~\ref{alg:efficientedgeselection}) takes time linear in the number of sites in
$\sigma_v$. Summing up over all nodes $v$ in a single level of $T$ we get that the total merging time per level is
$O(n)$, and
$O(n\log \Phi)$ for all levels.

To analyze the time taken by the edge selection steps
(line~\ref{line:edgeselection2} in
Algorithm~\ref{alg:efficientedgeselection}),
consider a particular pair $(\sigma, \tau) \in N$ for which
the algorithm runs the edge selection step.
By Lemma~\ref{lem:lowerenvelope}, if we charge $m_{\tau}$ by $O(1)$, each disk center in
$R_{\tau}$ by $O(\log n)$ and each active site in $\sigma \cap P$ by $O(1)$ then the total charges
cover the cost of the edge selection step for $(\sigma, \tau)$.
There are $O(n\log \Phi)$ nodes in  $T$ and therefore
 $O(n\log \Phi)$ cells $\tau$ in $\Q$.
By Lemma~\ref{lem:volume} each such cell $\tau$ participates
in an edge selection step of  $O(c^2) = O(1)$ pairs. So the total charges to
the site $m_{\tau}$ over all cells $\tau$, is $O(n\log \Phi)$.

By construction, each $p \in P$ is assigned to $O(1)$
sets $R_{\tau}$ and by Lemma~\ref{lem:volume} each $\tau$ participates
in an edge selection steps of  $O(c^2) = O(1)$ pairs.
It follows that the total charges to a site $p$ from edge selections steps
of pairs $(\sigma, \tau)$ such that $p\in R_\tau$ is $O(\log n)$.

Finally, each site is active for
$O(c^2) = O(1)$ pairs in $N$ at each of  $O(\log \Phi)$ levels.
So the total charges to a site $p$ from
edge selections steps
of pairs $(\sigma, \tau)$ such that $p$ is active in $\sigma \cap P$ is
$O( n \log \Phi)$. We conclude that the total running time of
all edge selection steps
is $O(n \log n + n \log \Phi) = O(n \log \Phi)$,
since $\log \Phi = \Omega(\log n)$.
\end{proof}

Theorem~\ref{thm:2dspannerSpread} follows by combining
Lemmas~\ref{lem:edgeapproximation} and~\ref{lem:runningtime}.

\subsection{From Bounded Spread to Bounded Radius Ratio}
\label{sec:spannerPsi}

Let $P\subset \R^2$ be a set of sites with radius ratio $\Psi$.
We extend our spanner construction
from Section~\ref{sec:spanner} such that the running time depends on $\Psi$, the ratio between the largest
to smallest radii, rather than
on the spread $\Phi$.  This is
a more general result as we may assume that $\Psi \leq 2 \Phi$ (see Section~\ref{sec:prelims}).
We prove the following theorem.
\begin{theorem}
\label{thm:2dspanner}
  Let $P$ be a set of $n$ sites in the plane
  with radius ratio $\Psi$. For any fixed $t > 1$,
  we can compute
  a $t$-spanner for the transmission graph $G$ of $P$ 
  in  $O(n(\log n + \log \Psi))$ time and  $O(n \log \Psi)$ space.
\end{theorem}
The main observation which we use is that
sites that are close together form a clique in $G$ and can be handled using
classic spanner constructions, while sites that are far away from each
other belong to distinct components of $G$ and can be dealt
with independently.

Given $t$, we pick sufficiently large constants $k=k(t)$ and $c=c(t)$ as specified in Section~\ref{sec:spanner}.
We scale the input such that the \emph{smallest radius} is
$c$. Let $M = c\Psi$ be the largest radius after we did the scaling.
First, we partition $P$ into sets that are far apart
and can be handled separately.

\begin{lemma}\label{lem:diamPartition}
  We can
  partition $P$ into sets $P_1, \dots, P_\ell$, such
  that each set $P_i$ has diameter $O(n\Psi)$ and
   for any $i \neq j$, no site of $P_i$ can
  reach a site of $P_j$ in $G$. Computing the partition takes
$O(n\log n)$ time and $O(n)$ space.
\end{lemma}

\begin{proof}
We assign to each site $p \in P$ an axis-parallel square
$S_p$ that is centered at $p$ and has side-length $2M$.
We define the intersection graph $G_S$
that has a vertex for each site in $P$, and
 an edge between two vertices $p$ and $q$ if and
only if $S_p \cap S_q \neq \emptyset$. ($G_S$ is undirected.)

If follows that if there
is no (undirected) path from $p$ to $q$ in $G_S$, then there
is no (directed) path from $p$ to $q$ in $G$.
We can compute the connected components of $G_S$ in $O(n \log n)$
time by sweeping the plane using a binary search tree~\cite{PreparataSh85}.
Let $P_1, \dots, P_\ell$ be the vertex sets of these connected
components. By construction, each
set of sites $P_i$ has diameter $O(nM)$ and
for any $i \neq j$, no site in $P_i$ can reach a site in $P_j$
in $G$.
\end{proof}

By Lemma~\ref{lem:diamPartition},  we  may assume that the diameter of our
input set $P$ is $O(n\Psi)$.
We compute a hierarchical decomposition
$T$ for $P$ as in Section~\ref{sec:spanner},
with a little twist as follows.
We
 translate $P$ so that it fits in a single grid cell $\sigma$ of
diameter $O(n\Psi)$.
Starting from $\sigma$, we recursively subdivide
each non-empty cell  into
four congruent cells of half the
diameter.
We do not subdivide cells of level $0$ whose  diameter
is $1$. We  partition all cells of a particular level in $O(n)$ time and
$O(n)$ space.

We construct a quadforest $T$ such that the roots of its trees
correspond to the non-empty cells of
 level $L = \lceil \log \Psi\rceil$ in our decomposition.
Each internal node of $T$ corresponds to a non-empty cell obtained when subdividing the cell of its parent.
It suffices to store only the lowest $L$ levels,
since larger cells cannot contribute
any edges to the spanner (as we will argue below). The forest  $T$ requires
$O(n \log \Psi)$ space and we compute it in
$O(n (\log n + \log \Psi))$ time.

We cannot derive from $T$ a $c$-separated annulus
decomposition for $G$ as we did in Section~\ref{sec:spanner}.
In particular a cell corresponding to a leaf of $T$ may now contain many sites
that are adjacent in $G$.
For edges induced by such pairs of sites we cannot satisfy Property (ii) of Definition~\ref{def:decomposition}.

We can (and do) derive from $T$ a \emph{partial $c$-separated
annulus decomposition $(\Q,N,R_\sigma)$} exactly as described in
Section \ref{sec:spanner} before Lemma \ref{lem:3.9}.
This decomposition satisfies Property (ii) of
Definition~\ref{def:decomposition} for all edges
$pq$ with $d(\sigma,\tau) \ge (c-2)$, where $\sigma$ and $\tau$ are
 the level 0 cells of $T$ containing $q$ and $p$, respectively.
The proof that Property (ii) of Definition~\ref{def:decomposition}
holds for these edges is the same as the proof of Lemma
\ref{lem:3.9}.
 In particular,
in the proof of Lemma \ref{lem:3.9}, we argue that pairs of cells at level $i$
guarantee Property (ii) of Definition~\ref{def:decomposition} for edges of length in $[c,2c)2^i$.
Since the edges of $G$ are of length at most
$M = c\Psi$,
the cells up to level $L = \lceil \log \Psi\rceil$ suffice to guarantee
Property (ii) of Definition~\ref{def:decomposition} for all
edges $pq$ with $d(\sigma,\tau) \ge (c-2)$.

We mark all sites of $P$ as active, and we run
Algorithm~\ref{alg:efficientedgeselection} of
Section~\ref{sec:spanner} using $T$ and the partial $c$-separated annulus
decomposition that we derived from it.
The resulting graph $H$ is not yet a $t$-spanner since
the decomposition  was only partial.

To make $H$ a spanner we add to it more edges that ``take care'' of the edges not
``covered'' by the $c$-separated annulus
decomposition. We consider each pair of level $0$ cells $\sigma$
and $\tau$ with
$d(\sigma, \tau) < c-2$. The set of sites $Q=(P\cap \sigma) \cup (P\cap \tau)$ form a clique, since the distance between
each pair of sites in $Q$ is no larger than $c$.
We compute a Euclidean $t$-spanner
for $Q$ of size $O(|Q|)$ in $O(|Q| \log |Q|)$ time~\cite{NarasimhanSmid07} and
for each (undirected) edge $pq$ of this spanner we add $pq$ and $qp$  to $H$.
As each site $p \in P$ participates in $O(c^2)$ such spanners,
we generate in total $O(n)$ edges in $O(n \log n)$ time.

We now prove that $H$ is indeed a $t$-spanner. The proof is analogous to the proof
of Lemma~\ref{lem:edgeapproximation}.

\begin{lemma}\label{lem:edgeapproximation_2}
For any $t > 1$, there are constants
$c=c(t)$ and $k=k(t)$ such that $H$ is a $t$-spanner for the transmission graph
$G$.
\end{lemma}
\begin{proof}
By construction, $H$ is a subgraph of $G$.
Let $pq$ be an edge of $G$,
and let $\sigma$ and $\tau$ be the level $0$ cells  with
$q \in \sigma$ and $p \in \tau$.
If $d(\sigma, \tau) < c-2$,
then the Euclidean $t$-spanner for
$\sigma$ and $\tau$ contains a path from $p$ to $q$
of length at most $t|pq|$.

For the remaining edges, the lemma is proved by induction
on the rank of the edges when we sort them by length,  as in
Lemma~\ref{lem:edgeapproximation}. The proof is almost
verbatim as before; we only comment on the base case.
Let $pq$ be the shortest edge in $G$.
If the endpoints $p$ and $q$ lie in
level 0 cells whose distance is less than $c-2$, we have
already argued that $H$ contains
an approximate path from $p$ to $q$. Otherwise, the same argument
as in Lemma~\ref{lem:edgeapproximation} applies, and the
algorithm includes $pq$ in $H$.
\end{proof}

Using Lemma~\ref{lem:edgeapproximation_2},
Theorem~\ref{thm:2dspanner} follows just as Theorem~\ref{thm:2dspannerSpread} in
Section~\ref{sec:spanner}. The analysis of the space and time
required by our construction is
exactly as in Lemma~\ref{lem:runningtime}, but now $T$
has $O(\log \Psi)$ levels.

\subsection{Spanners for Unbounded Spread and Radius Ratio}
\label{sec:spannerChan}

We eliminate
the dependency of our bounds on the radius ratio
at the expense of a more involved data structure and an additional
polylogarithmic factor in the running time.
Given $P \subset \R^2$ and the desired stretch factor $t > 1$, we choose
  appropriate  parameters $c = c(t)$ and $k(t)$ as in Section \ref{sec:spannerPsi} and
 rescale $P$ such that the distance between the closest pair of points in $P$ is
 $c + 2$.

To get the spanner of $G$ we compute a compressed quadtree $T$ for $P$.
A \emph{compressed quadtree} is a rooted tree in which each internal
node has degree $1$ or $4$. Each node $v$ is
associated with a cell $\sigma_v$ of a grid $\Q_i$.
If $v$ has degree $4$, then the  cells associated
of its children partition $\sigma_v$ into
$4$ congruent squares of half the diameter, and
at least two of them must be non-empty.
If $v$ has degree $1$, then the cell associated
with the only child $w$ of $v$ has diameter at most
$\diam(v)/4$ and  $(\sigma_v \setminus \sigma_w)\cap P = \emptyset$.
Each internal node of $T$ contains at least two sites
in its cell and each leaf at most one site.
For technical reasons we assume that the cell associated with a leaf $v$ has diameter $1$.
Since $v$ contains a single point $p$ we can artificially guarantee this by shrinking the 
cell associated with $v$ to the cell of
diameter one containing $p$.

Note that, in contrast with (uncompressed)  quadtrees, the diameter of
$\sigma_v$ may be smaller than $2^{L-i}$, where $i$ is the the distance of $v$ to the root and
$2^L$ is the diameter of the root.
A compressed quadtree for $P$ with $O(n)$ nodes
can be computed in $O(n\log n)$ time \cite{HarPeled11}.

To simplify the notation in the rest of this section, we write $\diam(v)$ instead of
$\diam(\sigma_v)$,  and for two nodes $v,w$, we
write $d(v,w)$ for $d(\sigma_v, \sigma_w)$.

Our approach is to use the algorithm from Section~\ref{sec:spanner}
on the compressed quadtree $T$.
One problem with this approach is
that the depth of
$T$ may be linear, so  considering all sites for incoming
edges at each level, as in Algorithm~\ref{alg:efficientedgeselection}, would be too expensive.
We tackle this difficulty by using
 Chan's dynamic nearest neighbor data structure
to speed up this stage. We achieve this speedup by reusing at a node $v$ the
largest structure among the structures at the children of $v$. The data structure of Chan has
 the following properties.

\begin{theorem}[Chan, Afshani and Chan, Chan and Tsakalidis, Kaplan 
et al~\cite{AfshaniCh09,Chan10, ChanTsakalidis15,KaplanMuRoSeSh17}]
\label{thm:chandynamicNN}
There exists a dynamic data structure that maintains a planar point set $S$
such that
\begin{enumerate}[(i)]
\item we can insert a point into $S$ in $O(\log^3 n)$  amortized time;
\item we can delete a point from $S$ in  $O(\log^5 n)$ amortized time; and
\item given a query point $q$, we can find the nearest neighbor
 of a query point $q$ in $S$ in  $O(\log^2 n)$ worst case time.
\end{enumerate}
The space requirement is $O(n)$.
\end{theorem}

We note that the history of Theorem~\ref{thm:chandynamicNN} is
a bit complicated: Chan's original paper~\cite{Chan10} describes
a \emph{randomized} data structure with $O(n \log\log n)$ space.
Afshahni and Chan~\cite{AfshaniCh09} describe a \emph{randomized}
three-dimensional range reporting structure that improves the space
to $O(n)$. Chan and Tsakalidis~\cite{ChanTsakalidis15} show how
to make both the dynamic nearest neighbor structure and the
range reporting structure deterministic. Kaplan et 
al~\cite{KaplanMuRoSeSh17} reduce the amortized deletion time
from $O(\log^6 n)$ to $O(\log^5 n)$, which gives the current
form of Theorem~\ref{thm:chandynamicNN}.

Another problem arises when we try to use the algorithm from Section~\ref{sec:spanner}
on the compressed quadtree $T$. We need to
 define an appropriate neighborhood
relation. The neighborhood relation from Section~\ref{sec:spanner} relied on the fact that
in a quadtree each point appears for every $i$ in the appropriate range in exactly one cell 
whose diameter is $2^i$. 
This is no longer the case in a compressed quadtree.

As in Section~\ref{sec:spanner},
the neighborhood relation $N$ which we define here would consist of pairs $(\sigma_v,\sigma_w)$
such that $\diam(v) = \diam(w)$ and
 $d(v,w) \in
[c-2, 2c) \diam(v)$.
The set $R_{\sigma_v}$ would consist of all sites
in $\sigma_v \cap P$ whose radius is in
$[c-2, 2(c+1))\diam(v)$, a slightly larger interval
than in the previous sections.
To make sure that $N$ and  $R_\sigma$ fulfill Property (ii) of
Definition~\ref{def:decomposition},
we insert $O(n)$ additional nodes into $T$ so that
$\Q$ contains the appropriate cells.
To find these nodes, we adapt the WSPD
algorithm of Callahan and Kosaraju \cite{CallahanKo95}.

\begin{lemma}
\label{lem:augmentingwithwspd}
Given a constant $c > 5$, we can in $O(n \log n)$ time insert $O(n)$ nodes into
$T$ so that
$\Q = \{ \sigma_v \mid v \in T\}$ with $N$ and $R_{\sigma}$ defined as stated above
is a $c$-separated annulus decomposition for $G$. In the same time, we
can compute $N$ and all sets $R_{\sigma}$.
\end{lemma}

\LinesNotNumbered
\begin{algorithm}[htbp]
call $\wspdone(r)$ on the root of $T$
\newline
\nl $\wspdone(v):$ \\
\nl \If{$v$ \textup{is a leaf}} {
\nl   \Return $\emptyset$ }
\nl \Else{
\nl   Return the union of $\wspdone(w)$ and $\wspdtwo(w_1,w_2)$ for all
  children $w$ and pairs of distinct children $w_1,w_2$ of $v$
}
\setcounter{AlgoLine}{0}
\nl $\wspdtwo(v,w):$
\newline
\nl \If{$d(v,w) \geq c\max \{\diam(v),\diam(w)\}$ } {
\nl   \Return $\{v,w\}$
}
\nl \ElseIf{ $\diam(v) \leq \diam(w)$}{
\nl  \Return the union of $\wspdtwo(v,u)$ for all children
  $u$ of $w$.
}
\nl \Else{
\nl\Return the union of $\wspdtwo(u,w)$ for all children
$u$ of $v$}
\caption{
 Computing a well-separated pair decomposition from a
compressed quadtree $T$. We scale the input such that the distance between the 
closest pair of points is $c+2$. This guarantees that when $v$ and $w$ are
both leaves, $\wspdtwo(v,w)$ returns $\{v,w\}$.}
\label{alg:wspd}
\end{algorithm}

\begin{proof}
First, we run the usual algorithm for finding a $c$-well-separated
pair decomposition on $T$~\cite{CallahanKo95};
see Algorithm~\ref{alg:wspd} for pseudocode. It is well
known~\cite{LofflerMu12} that the algorithm runs in $O(n)$ time and returns a
set $W$ of $O(n)$
pairs $\{v,w\}$ of nodes in $T$ such that
\begin{enumerate}[(a)]
\item for each two distinct sites $p$, $q$,
 there is exactly one
$\{v,w\} \in W$ with $q \in \sigma_v$, $p \in \sigma_w$;
\item for each $\{v,w\} \in W$, we have
$c\cdot \max\{\diam(v), \diam(w)\} \leq
d(v, w)$;
\item for every call $\wspdtwo(v,w)$,
$\max\{\diam(v), \diam(w)\} \leq
\min \{\diam(\overline{v}), \diam(\overline{w})\}$,
where
$\overline{v}$, $\overline{w}$ are the parents of $v$ and
$w$ in $T$;
\end{enumerate}

In particular, note that since we scaled $P$ such that the closest
pair has distance $c +2$, (b) is satisfied by any pair of (non-empty) cells
of $\Q_0$.

For each pair $\{v,w\} \in W$, we insert two nodes $v'$ and $w'$
into $T$ such that $\diam(v') = \diam(w')$
and such that $d(v',w')$ is approximately $c\cdot \diam(v')$.
Suppose that $\{v,w\}$ was generated through a call $\wspdtwo(v,\overline{w})$
 in Algorithm~\ref{alg:wspd} (the case that $\{v, w\}$ was generated through 
the call $\wspdtwo(\overline{v},w)$ is similar).
Let $r' = \min\{d(v,w)/c, \diam(\overline{w})\}$ and let $r$ be equal to $r'$
rounded down to the highest power of $2$.

\noindent
Observe that
\begin{equation}\label{equ:r_ub}
  r \leq \diam(\overline{w}) \leq \diam(\overline{v}),
\end{equation}
because $r \leq \diam(\overline{w})$ by definition, and
$\diam(\overline{w}) \leq \diam(\overline{v})$
by (c) and
our assumption that $\wspdtwo(v,\overline{w})$ was called.

\noindent
Furthermore, we have
\begin{equation}\label{equ:r_lb}
\max\{\diam(v), \diam(w)\} \le  r  .
\end{equation}
This follows from (c) if $r'=\diam(\overline{w})$ and from (b) if $r'=d(v, w)/c$ (recall 
that $\diam(v)$ and $\diam(w)$ are powers of two).

It follows from
 (\ref{equ:r_ub}) and (\ref{equ:r_lb}) that we can insert nodes $v'$ and $w'$ into $T$
between $v$ and $\overline{v}$ and between $w$ and $\overline{w}$, respectively,
such that $\diam(v') = \diam(w') = r$ and such
that
$\sigma_v \subseteq \sigma_{v'} \subseteq \sigma_{\overline{v}}$
and
$\sigma_w \subseteq \sigma_{w'} \subseteq \sigma_{\overline{w}}$.

We insert all these new nodes into $T$ efficiently by partitioning them according to
the parent-child pair in $T$ that they should be
inserted between.
We sort all the new nodes $x$ that should be inserted between each particular
parent-child pair $\overline{v},v$ by decreasing diameter and remove ``duplicate
nodes'': That is among each group of nodes of the same diameter we leave only one. 
Finally, we insert to $T$  a path consisting of the remaining nodes in order, making the first 
node on the path a child of
$\overline{v}$ and the last node on the path a parent of $v$.
It takes $O(n \log n)$ time to insert all the $O(n)$ new nodes.

To find the sets $R_{\sigma}$, we consider each site $p \in P$ and
 we identify
the nodes $v$ in $T$ such that $p \in R_{\sigma_v}$
in $O(\log n)$ time as follows. Since $c > 5$ there are
at most two integers $i$ such that
$r_p \in [c-2, 2(c+1))2^i$. For each such $i$, we identify (in $O(1)$ time) the
cell $\sigma \in \Q_i$ containing $p$
and then determine whether $\sigma$ is associated with a node $v$ in $T$.
The latter step requires  $O(\log n)$
time with an appropriate data structure. If indeed there is such a node $v$ we insert $p$ into
$R_{\sigma_v}$.
Thus, the total time we spend to find all sets $R_{\sigma}$ is $O(n \log n)$.
We compute the pairs in $N$ similarly also in
 $O(n \log n)$ time.

We now argue that this construction yields a $c$-separated annulus
decomposition for $P$. Property (i) of
Definition~\ref{def:decomposition} holds by construction.
To prove that Property (ii) of Definition
\ref{def:decomposition} holds consider some edge $pq$ in $G$.

Since $W$ is a $c$-WSPD, by (a) there
is a pair $\{v, w\} \in W$ with $q \in \sigma_v$ and $p \in \sigma_w$.
Suppose that $\{v, w\}$ was generated through the call $\wspdtwo(v,\overline{w})$.
Thus, we must have inserted nodes $v'$ and $w'$ into $T$ with
$\sigma_v \subseteq \sigma_{v'} \subseteq \sigma_{\overline{v}}$,
$\sigma_w \subseteq \sigma_{w'} \subseteq \sigma_{\overline{w}}$,
and with $\diam(v') = \diam(w') = r$.
Hence, $q \in \sigma_{v'}$ and $p \in \sigma_{w'}$.

We claim that $(\sigma_{v'},\sigma_{w'}) \in N$.
To prove this claim
 observe that since $r \leq d(v, w)/c$ it follows that
\begin{equation}\label{equ:dvw_lb}
  d(v', w') \geq d(v, w) - 2r \geq cr - 2r = (c-2)\diam(v'),
\end{equation}
Furthermore, if $r'=d(v,w)/c$, then
$d(v,w)/2c < r \leq d(v,w)/c$ and therefore
\begin{equation}\label{equ:dvw_ub1}
  d(v',w') \leq d(v, w) \leq 2c r.
\end{equation}

Since $\{v, w\}$ was generated through a call
$\wspdtwo(v,\overline{w})$ we know that $d(v, \overline{w}) \le c\diam(\overline{w}) $.
So if $r' = \diam(\overline{w})$ (implying $r=r'=\diam{(w')}=\diam(v')$) then we have
\begin{equation}\label{equ:dvw_ub2}
d(v',w') \leq d(v, \overline{w}) + \diam(v') \leq
(c + 1)r \leq  2cr .
\end{equation}

By (\ref{equ:dvw_lb}),(\ref{equ:dvw_ub1}) and (\ref{equ:dvw_ub2}),
we get $(\sigma_{v'},\sigma_{w'}) \in N$.
Finally, since $pq$ is an edge of $G$,
we have $r_p \geq d(v',w') \geq (c-2)\diam(w')$, by~(\ref{equ:dvw_lb}).
If $r_p < (c+1)\diam(w')$, then $p \in R_{\sigma_{w'}}$.
Otherwise let $m$ be the site in $\sigma_{w'} \cap P$ with the largest radius.
Then, $r_m \geq r_p \geq (c+1)\diam(w')$, so $D(m)$ contains
$\sigma_{v'}$ and thus $q$. This establishes Property (ii) of
Definition~\ref{def:decomposition}.
\end{proof}

\subparagraph*{Computing the Edges of $H$.}
As already mentioned, to construct the spanner $H \subseteq G$ for a
stretch factor $t > 1$,
we choose appropriate constants $k = k(t)$ and
$c = c(t)$, scale
 $P$ such
that the closest pair has distance $c+2$, and compute a compressed
quadtree $T$ for $P$. To obtain a
$c$-separated annulus decomposition
$(\Q, N, R_{\sigma})$ for $G$,
we augment $T$ with $O(n)$ nodes as described in the proof of Lemma~\ref{lem:augmentingwithwspd}.

We select the spanner edges for each cone  $C \in \C$ separately, as follows.
For each leaf $v$ of $T$,
we create a dynamic nearest neighbor (NN)
data structure $S_v$ as in Theorem~\ref{thm:chandynamicNN} containing initially the single point 
$p\in \sigma_v\cap P$.  We call a site $p$  \emph{active} if $p \in S_v$
for some node $v$ in $T$. So
initially, all sites of $P$ are active.
Then we process the nodes of $T$ in order of increasing diameter similarly to
Algorithm~\ref{alg:efficientedgeselection} of
Section~\ref{sec:spanner}.

Let $w$ be the child of $v$ such that $|S_w|$ is largest.
We generate $S_v$ from $S_w$ by inserting
into $S_w$ all the active sites of the children of $v$ other than $w$
 (we call this the \emph{preproccesing} step at $v$).
Then we use $S_v$ to do the edge selection for all $\tau \in N(\sigma_v)$
contained in $C_{\sigma_v}^2$; see
Algorithm~\ref{alg:NNedgeselection}. We take a site
$r \in R = R_{\tau} \cup \{m_{\tau}\}$ and repeatedly query $S_v$ for the site closest to $r$.
Let $q$ be the result. If $rq$ is an
edge in $G$, we add $rq$ to $H$, delete
$q$ from $S_v$, and do another query with $r$.
Otherwise, we continue with the next site of $R$,
until all of $R$ is processed. (This step is called the \emph{edge selection} step at $v$.)

\LinesNumbered
\begin{algorithm}[htb]
\tcp{preproccesing}
\label{line:NNpreproccesing}
 Let $w$ be the child of $v$ whose $S_w$
contains the most sites\\
Insert all active sites of each child $w' \neq w$ of $v$ into $S_w$\\
Set $S_v \gets S_w$\\
\ForEach{$\tau \in N(\sigma_v)$ \textnormal{contained in } $C_{\sigma_v}^2$}{
  \ForEach{$r \in R =  R_{\tau} \cup \{m_{\tau}\}$}{
    \tcp{edge selection}
    $q \gets \NN(v,r)$  \tcp{query $S_v$ with $r$}
    \label{line:firstNNquery}
    \While{$q \in D(r)$ \textnormal{and} $q \neq \emptyset$}{
    \label{line:NNwhile}
      add the edge $rq$ to $H$; delete $q$ from $S_v$; $q \gets \NN(v,r)$
    }
  }
  reinsert all deleted sites into $S_v$
}
delete all $q$ from $S_v$ for which at least one edge
$rq$ was found
\caption{Selecting incoming edges for the sites of a node $v$ and
 a cone $C$.}
\label{alg:NNedgeselection}
\end{algorithm}

The edges selected by Algorithm~\ref{alg:NNedgeselection}  have the
same properties as the edges selected by Algorithm~\ref{alg:edgeselection}.
Thus, by Lemma~\ref{lem:edgeapproximation} we obtain a $t$-spanner $H$.
Next, we analyze the running time.
\begin{lemma}
 \label{lem:runninglog6n}
 Algorithm~\ref{alg:NNedgeselection} has a total running time of
$O(n\log^5 n)$ and it requires  $O(n)$ space.
\end{lemma}
\begin{proof}
It takes $O(n\log n)$ to compute the compressed quadtree and to
find the neighboring pairs as in Lemma~\ref{lem:augmentingwithwspd}.
Initializing the nearest neighbor structures $S_v$ at the leaves $v$ takes $O(n)$ time.

Consider now the preprocessing phases at internal nodes $v$. That is the construction of
$S_v$ from  $S_w$ where $w$ is a child of $v$, by inserting into it the active sites from
structures $S_{w'}$ from the children $w' \neq w$ of $v$.
Since $S_w$ is the largest structure among the structures of the children of $v$, 
each time a site is inserted, the
size of the nearest neighbor structure that contains it increases by a factor
of at least two. Thus, each site is inserted $O(\log n)$ times.
By Theorem~\ref{thm:chandynamicNN} each  such insertion takes
 $O(\log^3 n)$ time. So the total time it takes to perform all these insertions
is $O(n\log^4 n)$.

For the edge selection,
consider two nodes $v$ and $w$ in $T$ whose cells are neighbors.
For each site $r$ in $R = R_{\sigma_w} \cup m_{\sigma_w}$,
we perform one nearest neighbor query at
line~\ref{line:firstNNquery} of Algorithm~\ref{alg:NNedgeselection} (the initial 
query with $r$).
We now evaluate what is the total time spent performing these \emph{initial queries}.

By Lemma~\ref{lem:volume} each cell  has
$O(c^2)$ neighbors so
each site $m_{\sigma_w}$ generates $O(c^2)$ queries. The total number of
sites $m_{\sigma_w}$ is equal to the number of nodes in $T$, which is $O(n)$.
Therefore the total number of initial nearest neighbor queries generated by sites $m_{\sigma_w}$ is $O(n)$.

Each site is assigned to $R_{\sigma_w}$ for at most two nodes $w$
and may generate $O(c^2)$ nearest neighbor queries when we process the neighboring cells of each such cell
${\sigma_w}$. Therefore the total number of initial nearest neighbor queries generated by sites in
sets $R_{\sigma_w}$ is also $O(n)$.

By Theorem~\ref{thm:chandynamicNN} the time it takes to perform a query is
$O(\log^2n)$ so the total time spent by initial queries is $O(n\log^2 n)$.

For each edge that we create in the while loop
of line~\ref{line:NNwhile},
we perform at most two deletions, one insertion and one
additional nearest neighbor query.
Since $H$ has $O(n)$ edges,
the total time required to perform these operations   is $O(n\log^5 n)$ 
by Theorem~\ref{thm:chandynamicNN}.

The total size of the compressed quadtree and of
the associated data structures is $O(n)$. Furthermore,
a dynamic nearest neighbor structure with $m$ elements
requires $O(m)$ space~\cite{Chan10}.
Thus, since at any time each site lies in at most one
dynamic nearest neighbor structure, the total space
requirement is $O(n)$.
\end{proof}

We conclude this section with the following theorem
that follows from  Lemma~\ref{lem:runninglog6n} and the discussion preceding it.

\begin{theorem}
\label{thm:2dspannerNN}
  Let $P \subset \R^2$ be an $n$-point set.
  For any $t > 1$, we can compute
  a $t$-spanner for the transmission graph $G$ of $P$ 
  in  $O(n\log^5 n)$ time
  and  $O(n)$ space.
\end{theorem}

\section{Applications}
We present two applications of our spanner construction.
We show how to use it to compute a breadth first search (BFS) tree 
from a particular vertex in a transmission graph, and
we show how to use it to extend a given reachability  data 
structure for additional queries specific to transmission
graphs. In both applications, we need to represent
the union of a set of disks in the plane (in our case these 
are the disks $D(p)$ for $p\in P$). It is well-known
that the boundary of this union has linear 
complexity~\cite{KedemLiPaSh86}. To represent
it algorithmically, we use the \emph{power diagram}, 
which is a weighted version of the Voronoi Diagram.
More specifically, the \emph{power distance} between a point $q$, 
and a disk with center $p$ and radius $r$, is $(d(p,r))^2 - r^2$. 
The power diagram partitions the plane into $n$ regions, such that all 
points in a specific region have the same closest disk in power distance.
The power diagram of a set of $n$ disks is of size $O(n)$ and 
can be constructed in $O(n\log n)$ time. If the power diagram 
is augmented with a point 
location structure, we can locate the disk $D$ that minimizes 
the power distance from a query point $q$ in $O(\log n)$ time. 
In particular we can determine in $O(\log n)$ time if $q$ is in 
the union of the disks  
by checking if $q\in D$~\cite{ImaiEtAl85,Kirckpatrick83}.

\subsection{From Spanners to BFS Trees}\label{sec:bfstree}

We show how to compute the BFS tree in a transmission graph $G$
from a given root $s \in P$ using the
spanner constructions from the previous section.
We adapt a technique that
Cabello and Jej\^ci\^c developed for unit-disk
graphs~\cite{CabelloJejcic15}.
Denote by $d_h(s,p)$ the BFS distance (also known as hop distance)
from $s$ to $p$ in $G$. Let $W_i \subseteq P$ be the sites $p\in P$ with $d_h(s,p) = i$.
Cabello and Jej\^ci\^c
used the Delaunay triangulation (DT) to efficiently
identify $W_{i+1}$, given $W_0, \dots, W_i$. We use
our $t$-spanner in a similar manner for transmission graphs.

\begin{lemma}
\label{lem:bfspathinspanner}
Let $t$ be small enough,  and
let $H$ be the $t$-spanner for $G$ as in
Theorem~\ref{thm:2dspannerSpread}, \ref{thm:2dspanner} or~\ref{thm:2dspannerNN}.
Let $v \in W_{i+1}$, for some $i \geq 1$.
Then, there is a site $u \in W_i$ and a path
$u = q_\ell,\dots,q_1 = v$ in $H$ with $d_h(s,q_j) = i+1$ for $j = 1, \dots, \ell$.
\end{lemma}

\begin{proof}
We focus on the spanner from Theorem~\ref{thm:2dspanner}, since it has
the most complicated structure. The proof for the other constructions is similar and simpler.

Since $v \in W_{i+1}$, there is a $w \in W_i$ with $v \in D(w)$.
If $H$ contains the edge $wv$,
the claim follows by setting $u = q_2= w$ and $q_1=v$.
Otherwise, we construct the path backwards from $v$ (see
Figure~\ref{fig:bfspathinspanner}).
Suppose we have already constructed a sequence
$v = q_1, q_2, \dots, q_k$
of sites in $P$ such that
(i) for $j = 1, \dots, k-1$, $q_{j+1}q_j$ is an edge of
$H$; (ii) for $j = 1, \dots k$, we have $q_j \in D(w)$ and
$d_h(s, q_j) = i+1$; and
(iii) for $j = 1, \dots, k-1$, $|wq_{j+1}| < |wq_{j}|$.
We begin with the sequence $q_1 = v$ satisfying the invariant.

\begin{figure}[htb]
\begin{center}
\includegraphics[scale=0.8]{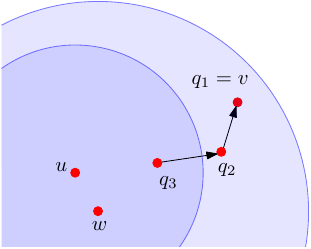}
\end{center}
\caption{The partial path constructed backwards from $v$. Setting $q_4=u$ will
complete it.}
\label{fig:bfspathinspanner}
\end{figure}

Let $c$ be the constant from the spanner construction of
Section~\ref{sec:spannerPsi}, and
recall that we scale $P$ such that the smallest radius is $c$.
Suppose that we have $q_{1},\dots, q_{k}$ and that $wq_k$
is not an edge of $H$ (otherwise  we could finish by setting $u = w$).
Let $\sigma, \tau \in \Q_0$ be the cells such that $w \in \tau$ and
$q_k \in \sigma$. We distinguish two cases, depending on
$d(\sigma,\tau)$, and we either show how to find $u$
to complete the path from $u$ to $v$ or how to choose $q_{k+1}$.

Case 1: $d(\sigma,\tau) < c-2$.
Let $Q = (P \cap \sigma) \cup (P \cap \tau)$.
We have that  $w, q_k \in Q$.
The algorithm of Section~\ref{sec:spannerPsi}
constructs a Euclidean spanner for $Q$ and adds its edges to $H$.
In particular, there is a directed path $\pi$ from $w$ to $q_k$
that uses
only sites of $Q$. By construction, the pairwise distances between the sites of $Q$ are all
at most $c$. Thus, for each $p \in Q$ we have $p \in D(w)$ and
$q_k \in D(p)$,
and therefore $i \leq d_h(s,p) \leq i+1$.
We set  $u$ be the last site of $\pi$
with $d_h(s,u) = i$.
To obtain the desired path from $u$ to $v$
we take the subpath of $\pi$ starting at $u$ and concatenate it to the
the partial path $q_k,\dots,q_1=v$.

Case 2: $d(\sigma,\tau) \geq c-2$.
Since $wq_k$ is not an edge of $H$, by Lemma~\ref{lem:shorteredge}
there exists an edge
$rq_k$ in $H$ with $|wr| < |wq_k|$. We
set $q_{k+1} = r$.
Since $q_k \in D(w)$, we have $q_{k+1} \in D(w)$
and $i \leq d_h(s,q_{k+1}) \leq i+1$. If $d_h(s,q_k) = i$, we set
$u = q_{k+1}$ and are done. Otherwise, $q_{k+1}$
satisfies properties (i)--(iii) and we continue to extend the path.

Since the distance to $w$ decreases in each step and
since $P$ is finite, this process eventually stops and the lemma follows.
\end{proof}

\LinesNumbered
\begin{algorithm}[hbt]
$W_0 \leftarrow \{s\}$; $\dist[s]=0$; $\pi[s]=s$; $i = 0$;
and, for $p \in P \setminus \{s\}$, $\dist[p] = \infty$ and $\pi[p] = \NIL$\\
\While{$W_i \neq \emptyset$} {
  compute power diagram with point location
  structure $\PD_i$ of $W_i$\\
  queue $Q \leftarrow W_i$
  \label{step:fillingQ} ; $W_{i+1} \leftarrow \emptyset$ \\
  \While{$Q \neq \emptyset$} {
      $p\gets \text{dequeue}(Q)$\\
      \ForEach{\textnormal{edge} $pq$ of $H$} {
        \label{step:forloop}
        $u \gets \PD_i(q)$ \tcp{query $\PD_i$ with $q$, $D(u)$ minimizes the power distance from $q$}
        \label{step:querypd}
        \If{$q \in D(u)$ \textnormal{and} $\dist[q] = \infty$} {
          \label{step:validvertices}
          $\text{enqueue}(Q,q)$; $\dist[q] = i+1$;
          $\pi[q] = u$; add $q$ to $W_{i+1}$
          \label{step:endif}
        }
      }
  }
   $i \gets i+1$
}
\caption{Computing the BFS tree for $G$ with root $s$ using the spanner $H$.}
\label{alg:bfstree}
\end{algorithm}

The BFS tree for $s$ is computed iteratively;
see Algorithm~\ref{alg:bfstree} for pseudocode.
Initially, we set $W_0 = \{s\}$.
Now assume we have computed $W_0,\ldots,W_i$.
By Lemma~\ref{lem:bfspathinspanner}, all sites in $W_{i+1}$
can be reached from $W_i$ in the subgraph of $H$
induced by $W_i \cup W_{i+1}$.
Thus, we can compute  $W_{i+1}$ by running  a BFS search in $H$ from the points of $W_i$ using a queue $Q$. Every time
we encounter a new vertex $q$, we check if it lies in
a disk around a site of $W_i$, and is not yet in the BFS tree for $s$.
If so, we add $q$ to $W_{i+1}$ and to $Q$.
 Otherwise, we discard $q$.
To test whether $q$ lies in a disk of $W_i$, we compute a power diagram for
$W_i$ in time $O(|W_i|\log |W_i|)$ and query it with $q$.

A site $p$ at level $i$ is traversed by at most two BFS searches in $H$. In the first search we discover that $p$
is in $W_i$, and in the second search $p$ is a starting point --- this is  the search to discover  $W_{i+1}$.
It follows that an edge $pq$ of $H$ is considered twice
by Algorithm~\ref{alg:bfstree}. Each time we consider the edge
$pq$ we spend $O(\log n)$ time for  querying a power diagram with $q$.
 Since $H$ is sparse, the total time required is
$O(n\log n)$. This establishes the following theorem.
\begin{theorem}
Let $P \subset \R^2$ be a set of $n$ points.
 Given a spanner $H$ for the transmission graph $G$ of $P$ 
 as in Theorem~\ref{thm:2dspannerSpread},
 Theorem~\ref{thm:2dspanner}, or Theorem~\ref{thm:2dspannerNN}, we can
 compute in $O(n\log n)$ additional time a BFS tree in $G$ 
 rooted at any given site $s \in P$.
\end{theorem}

\subsection{Geometric Reachability Oracles}
\label{sec:oracles}

Let $G$ be a directed graph. If there is a directed path from a vertex $s$
to a vertex $t$ in $G$, we say $s$  \emph{can reach} $t$ (in $G$).
A \emph{reachability oracle} for a graph $G$ is a data structure
that can answer efficiently for any given pair $s$, $t$ of vertices of $G$ whether $s$ can reach $t$.
Reachability oracles  have been studied
extensively over the last decades (see, e.g.,~\cite{Holm2015,Thorup04} and
the references therein).

When $G$ is a transmission graph
we are interested in
 a more
general type of reachability  query where the target $t$ is not necessarily a vertex of $G$,
but an arbitrary point in the plane. We say that a site $s$ can
reach a point $t \in \R^2$ if there is a 
site $q$ in $G$ such that $t \in D(q)$
and such that $s$ can reach $q$ in $G$. We call a data structure that supports this type of
queries a \emph{geometric reachability oracle}. We can use our
spanner construction from Theorem~\ref{thm:2dspanner} to extend any
reachability oracle for a transmission graph to a geometric reachability oracle
with a small overhead in space and query time.
More precisely, we prove the following theorem.
\begin{theorem}
\label{thm:geometricreachability}
Let $P$ be a set of $n$ points in the plane with radius
ratio $\Psi$.
Given a reachability oracle for the transmission graph $G$ of $P$ 
that requires $S(n)$ space and
has query time $Q(n)$, we can obtain in $O(n \log n \log \Psi)$ time a
geometric reachability oracle for $G$ that requires 
$S(n) + O(n \log \Psi)$ space and
can answer a query in $O(Q(n) + \log n \log \Psi)$ time.
\end{theorem}

Given a query $s,t$ with a target $t \in \R^2$, our strategy is to find a small
subset $Q \subseteq P$  such that
for each $q \in Q$,  $t\in D(q)$, and $Q$ ``covers the space around $t$'' in the following
sense. For any disk $D(p)$ such that  $t \in D(p)$ there is a site $q
\in Q$ with $q \in D(p)$. In particular the edge $pq$ is in $G$.

Such a set $Q$ satisfies that $s$ can reach $t$ if and only if $s$ can reach some site $q\in Q$.
Once we have computed  $Q$ we
decide whether $s$ can reach $t$ by querying the
given reachability oracle with $s,q$ for all $q \in Q$.
The answer is positive if and only if it is positive for at least
one site $q \in Q$.

In what follows, we construct a data
structure of size  $O(n \log \Psi )$ that allows to find such a set $Q$ of size $O(1)$ in  $O(\log n
\log \Psi)$ time. Theorem~\ref{thm:geometricreachability} is then immediate.

\paragraph*{The Data Structure.}
We compute a $2$-spanner $H$ for $G$ as in Theorem~\ref{thm:2dspanner}. Let $k$ (the number
of cones) and $c$ (the separation parameter) be the
two constants used by the construction of $H$, and recall that we scaled $P$
such that
the smallest radius of a site in $P$ is $c$.
Let $T$ be the quadforest used by the
construction of $H$. The trees in  $T$ have depth $O(\log \Psi)$
and each node $v \in T$ corresponds
to a grid cell $\sigma_v$ from some grid $\Q_i$, $i \geq 0$. Our data structure
is obtained by augmenting each node $v \in T$
by a power diagram $\PD_{\sigma_v}$ for the sites in
$\sigma_v \cap P$, together with a point location data structure. This requires
$O(|\sigma_v \cap P|)$ space and  $O(|\sigma_v \cap P|\log |\sigma_v \cap P|)$
time~\cite{ImaiEtAl85,Kirckpatrick83} for each $v$. Since any site of $P$ is in
$O(\log \Psi)$ cells of $T$, we need $O(n \log\Psi)$ space and $O(n \log n \log
\Psi)$ time in total.
\begin{algorithm}[htb]
$L \gets $ depth of $T$\\
\For{$i=0,\dots, L$} {
  $\sigma \gets $ cell of $\Q_i$ with $t \in \sigma$\\
  \ForEach{$\tau \in N(\sigma)$ \textnormal{contained in}
$C_{\sigma}^2$}{
      $q \gets \PD_\tau(t)$ \tcp{query $\PD_\tau$ with $t$}
      if $t \in D(q)$, add $q$ to $Q$
    }
    Stop if at least one $q$ was added to $Q$
}
\caption{Query Algorithm for a cone $C$ and a point $t$.}
\label{alg:querycoverset}
\end{algorithm}

\paragraph*{Performing a Query.}
Let a query point $t \in \R^2$ be given. Let $\sigma$ be the cell in $\Q_0$ that
contains $t$.
To find $Q$, we first traverse all non-empty cells
$\tau \in \Q_0$ with $d(\sigma,\tau) \leq c - 2$.
From each such cell $\tau$,  if there exists a site $q \in \tau \cap P$ such that
$t \in D(q)$ then we add one, arbitrary, such site to $Q$.
To determine if such a site exists, and to find one if it exists, we query $\PD_\tau$ with $t$.
Second, we go through all cones $C \in \C$, and we run
Algorithm~\ref{alg:querycoverset} with $C$ and $t$ to find the remaining sites
for $Q$.
Algorithm~\ref{alg:querycoverset} is similar to
Algorithms~\ref{alg:edgeselection} and~\ref{alg:efficientedgeselection}, and
computes the incoming edge of $t$ if it would have been inserted into the spanner.
 We go through the grids at all
levels of $T$. For each level we consider the cell
$\sigma$ that contains $t$ and for each cell
$\tau \in N(\sigma)$ that is contained in
$C_{\sigma}^2$ we select a site with an edge to $t$ if there is one.
Lemma~\ref{lem:shorteredge} holds for the incoming edges of $t$ and using
this fact, we can prove
that our data structure has the desired properties.
\begin{lemma}
\label{lem:coverset}
Let $P$ be a set of $n$ points in the plane with radius ratio $\Psi$. We can
construct in $O(n\log n \log \Psi)$ time a data structure that finds for any
given query point $t \in \R^2$ a set $Q \subseteq P$ such that $|Q|= O(1)$ and
for any site $p \in P$, if $t \in D(p)$ we have that $D(p) \cap Q \neq
\emptyset$.
The query time is $O(\log n \log \Psi)$ and the space requirement is $O(n
\log \Psi)$.
\end{lemma}

\begin{proof}
The construction time and the space requirement are immediate.
For the query time recall that $T$ has depth $O(\log \Psi)$
and by Lemma~\ref{lem:volume}, at each level we make $O(c^2)$ queries to the
power diagrams. It follows that
 it takes
$O(\log n\log \Psi)$ time to compute $Q$.

By construction, $Q$ has size $O(1)$. Indeed, at the first step, we add at most one
site for every cell of distance at most $c-2$ from $\sigma$, and there are
$O(c^2)$ such cells. In the second step, for each cone, we only add sites from $O(c^2)$ cells at
one level of $T$.

Now let $p \in P$ be a site with $t \in D(p)$. It remains to show that
$D(p) \cap Q \neq \emptyset$. If $p \in Q$, we are done. If not, we let
$\sigma$ and $\tau$ be the cells in $\Q_0$ with $t \in \sigma$ and $p \in
\tau$.
If $d(\sigma,\tau) \leq c - 2$ then  there must be a site $q \in \tau \cap Q$.
Since $\diam(\tau) = 1$ and $r_p \geq c$, we have $q \in D(p)$.
If $d(\sigma,\tau) > c-2$
then since $pt$ is an edge in $G$ that is not selected
by Algorithm~\ref{alg:querycoverset},  Lemma~\ref{lem:shorteredge} guarantees that there is
an edge $qt$ with $q \in Q$ and $|pq| < |pt|$. Since $t \in D(p)$ we also
have $q \in D(p)$. This finishes the proof.
\end{proof}

\section{Conclusion}
We have described the first construction of spanners for
transmissions graphs that runs in near-linear time, and we demonstrated
its usefulness by describing two applications. Our techniques are quite general, and we
expect that they will be applicable in similar settings.
For example, in an ongoing work we consider how to extend our results
to (undirected) disk intersection graphs. This would significantly improve
the bounds of F\"urer and Kasiviswanathan~\cite{FuererKasiviswanathan12}.

Our most general spanner construction requires
a dynamic data structure for planar Euclidean nearest neighbors.
It is an interesting
challenge to find a simpler solution that possibly avoids the need for such a structure.

Finally, we believe that our work indicates that transmission graphs
constitute an interesting and fruitful model of geometric graphs
worthy of further investigation. In a companion paper~\cite{KaplanEtAl15b},
we consider several questions concerning reachability
in transmission graphs. In particular, we describe several
constructions of reachability oracles for transmission graphs 
(see Section~\ref{sec:oracles}),
providing many opportunities to apply 
Theorem~\ref{thm:geometricreachability}.
Also, in this context our spanner construction plays a crucial role
in obtaining fast preprocessing algorithms.

\paragraph{Acknowledgments.}
We like to thank Paz Carmi and G\"unter Rote for valuable comments.
We also thank the anonymous referees for their careful reading of
the paper and for their insightful suggestions, and in particular for
pointing out the problem of geometric reachability queries as described in
Section~\ref{sec:oracles}.

\bibliographystyle{abbrv}
\bibliography{literature}

\newcommand{\SortNoop}[1]{}
\begin{thebibliography}{10}

\bibitem{AfshaniCh09}
P.~Afshani and T.~M. Chan.
\newblock Optimal halfspace range reporting in three dimensions.
\newblock In {\em Proc. 20th Annu. ACM-SIAM Sympos. Discrete Algorithms
  (SODA)}, pages 180--186, 2009.

\bibitem{4M}
M.~{\SortNoop{Berg}}de~Berg, O.~Cheong, M.~van Kreveld, and M.~H. Overmars.
\newblock {\em Computational Geometry: Algorithms and Applications}.
\newblock Springer-Verlag, 3rd edition, 2008.

\bibitem{BoseEtAl12}
P.~Bose, M.~Damian, K.~Dou{\"i}eb, J.~O'Rourke, B.~Seamone, M.~H.~M. Smid, and
  S.~Wuhrer.
\newblock $\pi/2$-angle {Y}ao graphs are spanners.
\newblock {\em Internat. J. Comput. Geom. Appl.}, 22(1):61--82, 2012.

\bibitem{Boukerche08}
A.~Boukerche.
\newblock {\em Algorithms and Protocols for Wireless Sensor Networks}.
\newblock Wiley Series on Parallel and Distributed Computing. Wiley-IEEE Press,
  1st edition, 2008.

\bibitem{CabelloJejcic15}
S.~Cabello and M.~Jej\^ci\^c.
\newblock Shortest paths in intersection graphs of unit disks.
\newblock {\em Comput. Geom.}, 48(4):360--367, 2015.

\bibitem{CallahanKo95}
P.~B. Callahan and S.~R. Kosaraju.
\newblock A decomposition of multidimensional point sets with applications to
  $k$-nearest-neighbors and $n$-body potential fields.
\newblock {\em J. ACM}, 42(1):67--90, 1995.

\bibitem{Carmi14}
P.~Carmi, 2014.
\newblock personal communication.

\bibitem{Chan10}
T.~M. Chan.
\newblock A dynamic data structure for 3-{D} convex hulls and 2-{D} nearest
  neighbor queries.
\newblock {\em J. ACM}, 57(3):Art. 16, 15, 2010.

\bibitem{ChanTsakalidis15}
T.~M. Chan and K.~A. Tsakalidis.
\newblock Optimal deterministic algorithms for 2-d and 3-d shallow cuttings.
\newblock In {\em Proc. 31st Int. Sympos. Comput. Geom. (SoCG)}, pages
  719--732, 2015.

\bibitem{ChangEtAl90}
M.~S. Chang, N.~F. Huang, and C.~Y. Tang.
\newblock An optimal algorithm for constructing oriented {V}oronoi diagrams and
  geographic neighborhood graphs.
\newblock {\em Inform. Process. Lett.}, 35(5):255--260, 1990.

\bibitem{Clark90}
B.~N. Clark, C.~J. Colbourn, and D.~S. Johnson.
\newblock Unit disk graphs.
\newblock {\em Discrete Math.}, 86(1-3):165--177, 1990.

\bibitem{FuererKasiviswanathan12}
M.~F{\"u}rer and S.~P. Kasiviswanathan.
\newblock Spanners for geometric intersection graphs with applications.
\newblock {\em J. Comput. Geom.}, 3(1):31--64, 2012.

\bibitem{HarPeled11}
S.~Har-Peled.
\newblock {\em Geometric Approximation Algorithms}.
\newblock American Mathematical Society, 2011.

\bibitem{Holm2015}
J.~Holm, E.~Rotenberg, and M.~Thorup.
\newblock Planar reachability in linear space and constant time.
\newblock In {\em Proc. 56th Annu. IEEE Sympos. Found. Comput. Sci. (FOCS)},
  pages 370--389, 2015.

\bibitem{ImaiEtAl85}
H.~Imai, M.~Iri, and K.~Murota.
\newblock {V}oronoi diagram in the {L}aguerre geometry and its applications.
\newblock {\em SIAM J. Comput.}, 14(1):93--105, 1985.

\bibitem{KaplanMuRoSe15}
H.~Kaplan, W.~Mulzer, L.~Roditty, and P.~Seiferth.
\newblock Spanners and reachability oracles for directed transmission graphs.
\newblock In {\em Proc. 31st Int. Sympos. Comput. Geom. (SoCG)}, pages
  156--170, 2015.

\bibitem{KaplanEtAl15b}
H.~Kaplan, W.~Mulzer, L.~Roditty, and P.~Seiferth.
\newblock Reachability oracles for directed transmission graphs.
\newblock \texttt{arXiv:1601.07797}, 2016.

\bibitem{KaplanMuRoSeSh17}
H.~Kaplan, W.~Mulzer, L.~Roditty, P.~Seiferth, and M.~Sharir.
\newblock Dynamic planar {V}oronoi diagrams for general distance functions and
  their algorithmic applications.
\newblock In {\em Proc. 28th Annu. ACM-SIAM Sympos. Discrete Algorithms
  (SODA)}, pages 2495--2504, 2017.

\bibitem{KedemLiPaSh86}
K.~Kedem, R.~Livne, J.~Pach, and M.~Sharir.
\newblock On the union of {J}ordan regions and collision-free translational
  motion amidst polygonal obstacles.
\newblock {\em Discrete Comput. Geom.}, 1:59--70, 1986.

\bibitem{Kirckpatrick83}
D.~Kirkpatrick.
\newblock Optimal search in planar subdivisions.
\newblock {\em SIAM J. Comput.}, 12(1):28--35, 1983.

\bibitem{LofflerMu12}
M.~L{\"o}ffler and W.~Mulzer.
\newblock Triangulating the square and squaring the triangle: quadtrees and
  {D}elaunay triangulations are equivalent.
\newblock {\em SIAM J. Comput.}, 41(4):941--974, 2012.

\bibitem{NarasimhanSmid07}
G.~Narasimhan and M.~H.~M. Smid.
\newblock {\em Geometric spanner networks}.
\newblock Cambridge University Press, 2007.

\bibitem{PelegRoditty10}
D.~Peleg and L.~Roditty.
\newblock Localized spanner construction for ad hoc networks with variable
  transmission range.
\newblock {\em ACM Transactions on Sensor Networks (TOSN)}, 7(3):25:1--25:14,
  2010.

\bibitem{PreparataSh85}
F.~P. Preparata and M.~I. Shamos.
\newblock {\em Computational geometry. An introduction.}
\newblock Springer-Verlag, 1985.

\bibitem{AgarwalSharir96}
M.~Sharir and P.~K. Agarwal.
\newblock {\em {D}avenport-{S}chinzel sequences and their geometric
  applications}.
\newblock Cambridge University Press, 1996.

\bibitem{Thorup04}
M.~Thorup.
\newblock Compact oracles for reachability and approximate distances in planar
  digraphs.
\newblock {\em J. ACM}, 51(6):993--1024, 2004.

\bibitem{RickenbachEtAl09}
P.~von Rickenbach, R.~Wattenhofer, and A.~Zollinger.
\newblock Algorithmic models of interference in wireless ad hoc and sensor
  networks.
\newblock {\em IEEE/ACM Transactions on Networking}, 17(1):172--185, 2009.

\bibitem{Yao82}
A.~C.-C. Yao.
\newblock On constructing minimum spanning trees in $k$-dimensional spaces and
  related problems.
\newblock {\em SIAM J. Comput.}, 11(4):721--736, 1982.

\end{thebibliography}
\end{document}